\documentclass[pdflatex,sn-mathphys-num]{sn-jnl}% Basic Springer 

\usepackage{graphicx}%
\usepackage{multirow}%
\usepackage{amsmath,amssymb,amsfonts}%
\usepackage{amsthm}%
\usepackage{mathrsfs}%
\usepackage[title]{appendix}%
\usepackage{xcolor}%
\usepackage{textcomp}%
\usepackage{manyfoot}%
\usepackage{booktabs}%
\usepackage{listings}%
\usepackage[ruled]{algorithm2e}

\usepackage{algorithmic}

\usepackage{diagbox}

\usepackage{subcaption}   % 导言区加载

\theoremstyle{thmstyleone}%
\newtheorem{theorem}{Theorem}%  meant for continuous numbers
\newtheorem{proposition}[theorem]{Proposition}% 

\theoremstyle{thmstyletwo}%
\newtheorem{remark}{Remark}%

\theoremstyle{thmstylethree}%
\newtheorem{definition}{Definition}%

\newtheorem{lemma}{Lemma}

\begin{document}

% \title[Group Stepdown SLOPE for Controlled Feature Selection]{Group Stepdown SLOPE for Controlled Feature Selection with Group Structures}

\title{Beyond False Discovery Rate: A Stepdown Group SLOPE Approach for Grouped Variable Selection}

\author[1]{Xuelin Zhang}\email{zhangxuelin@webmail.hzau.edu.cn}

\author[1]{Jingxuan Liang}\email{1908103270@qq.com}

\author[1]{Xinyue Liu}\email{lxy36@webmail.hzau.edu.cn}

\author[1]{Hong Chen}\email{chenh@mail.hzau.edu.cn}

\author[1]{Biqin Song}\email{biqin.song@mail.hzau.edu.cn}
% \equalcont{Corresponding author.}

\affil[1]{\orgdiv{College of Informatics}, \orgname{Huazhong Agricultural University}, \orgaddress{\city{Wuhan}, \postcode{430070}, \state{Hubei Province}, \country{China}}}

% \affil[2]{\orgdiv{Engineering Research Center of Intelligent Technology for Agriculture}, \orgname{Ministry of Education}, \orgaddress{\city{Wuhan}, \postcode{430070}, \state{Hubei Province}, \country{China}}}

\abstract{
High-dimensional feature selection is routinely required to balance statistical power with strict control of multiple-error metrics such as the k-Family-Wise Error Rate (k-FWER) and the False Discovery Proportion (FDP), yet some existing frameworks are confined to the narrower goal of controlling the expected False Discovery Rate (FDR) and can not exploit the group-structure of the covariates, such as Sorted L-One Penalized Estimation (SLOPE). We introduce the Group Stepdown SLOPE, a unified optimization procedure which is capable of embedding the Lehmann–Romano stepdown rules into SLOPE to achieve finite-sample guarantees under k-FWER and FDP thresholds. Specifically, we derive closed-form regularization sequences under orthogonal designs that provably bound k-FWER and FDP at user-specified levels, and extend these results to grouped settings via gk-SLOPE and gF-SLOPE, which control the analogous group-level errors gk-FWER and gFDP. For non-orthogonal general designs, we provide a calibrated data-driven sequence inspired by Gaussian approximation and Monte-Carlo correction, preserving convexity and scalability. Extensive simulations are conducted across sparse, correlated, and group-structured regimes. Empirical results corroborate our theoretical findings that the proposed methods achieve nominal error control, while yielding markedly higher power than competing stepdown procedures, thereby confirming the practical value of the theoretical advances. 
% The implementation is publicly available at \url{https://github.com/zhangxuelincode/sdGSLOPE}.
}


\maketitle

\section{Introduction}
In specific statistical applications of high-dimensional signal processing, only a subset of features is crucial for predicting outcomes, and relevant variables can be grouped, where all variables in each subgroup will be selected or discarded. Thus, it's vital to identify critical features from sparse data to enhance computational efficiency and provide interpretable predictions \citep{wu2023automated,su2025toward}.
Feature selection for a single regressor or at the group level aims to identify informative features or feature subsets from high-dimensional empirical observations, a key research field in machine learning. Typical feature selection methods include sparse linear models (e.g., Lasso \citep{lasso}, Group Lasso \citep{yuan2006model}, Sparse Group Lasso \citep{sglasso-NIPS19}), sparse additive models (e.g., SpAM \citep{spam}, GroupSAM \citep{chenh-nips17}, SpMAM \citep{chenh-tnnls2020} and NAMs \citep{zhang2024gaussian,zhu2024neural,zhang2025interpretable}), tree-based models (e.g., random forest \citep{randomforest}), and sparse neural networks (e.g., LassoNet \citep{lemhadri-2021lassonet}, Deep feature selection \citep{taherkhani2018deep}, Deep-Pink \citep{deeppink}, multi-scale feature selection network \citep{li2024multi}). Nevertheless, although most of these approaches exhibit good empirical performance in feature selection, the theoretical guarantees regarding the false discovery rate (FDR) are often lacking, which is extremely important.

Following this line of research, controlled feature selection further improves selection quality while guaranteeing low FDR, which has recently attracted increasing attention due to its wide applications, e.g., in bioinformatics and biomedical processing \citep{aggarwal-2016biofdr, yu-2021biofdr}. 
There are mainly three branches of learning systems for controlled feature selection: the multiple hypothesis test \citep{benjamini-1995bh, ferreira-2006bh, kFWER2, stepup}, the knockoffs filter \citep{barber-2015knockoffs, Cands2016PanningFG, barber-2020knockoffs,romano-2020knockoffs}, and the Sorted L-One Penalized Estimation (SLOPE) \citep{slope, slopestatics, Groupslope}. As a classic strategy for feature selection, the Benjamini-Hochberg (BH) procedure is formulated by jointly considering p-values from multiple hypothesis tests \citep{benjamini-1995bh}. 
Despite this procedure enjoying attractive theoretical properties for FDR control, it may face computational challenges in nonlinear and complex regression estimation \citep{javanmard2019false}. As a novel feature filter scheme, knockoff inference has solid theoretical foundations and demonstrates competitive performance in real-world applications \citep{barber-2015knockoffs,barber-2020knockoffs,zhao-2022error,yu-2021biofdr}. In particular, an error-based knockoffs inference framework is formulated in \citep{zhao-2022error} to further realize controlled feature selection from the perspectives of the probability of $k$ or more false rejections ($k$-FWER) and the false discovery proportion (FDP). Different from screening out the active feature with the help of knockoff features, SLOPE focuses on the regularization design for sparse feature selection, which adaptively imposes a non-increasing sequence of tuning parameters on the sorted $\ell_1$ penalties \citep{slope, Groupslope, abslope}. This is further extended to select feature subsets by exploiting group-level sparsity regularization, referred to as g-SLOPE \citep{Groupslope}.
Although there has been rapid progress on the optimization algorithm of SLOPE \citep{slope, Groupslope} and its theoretical properties \citep{slopestatics}, most existing works on SLOPE and g-SLOPE are limited to FDR control alone. 
Motivated by this limitation, we further examine the grouping structure and develop new Group SLOPE approaches for controlled group-wise feature selection under $k$-FWER and FDP criteria. Moreover, we conduct sufficient experiments under different configurations of orthogonal or Gaussian designs, levels of sparsity and numbers of relevant features (and groups) to validate the efficiency of the proposed new approaches for group-wise feature selection.

\begin{table*}[!t]
\centering
\caption{Algorithmic properties ($\surd$ - has the given information, $\times$ - hasn’t the given information).}
\resizebox{\linewidth}{!}{
\begin{tabular}{l|ccccccc}
\hline
\diagbox{Algorithms}{Properties}  & FDR$^1$ Control & FDP$^1$ Control & $k$-FWER$^1$ Control & Group-wise\\ \hline
SLOPE\citep{slope,slopestatics,larsson2020strong,dupuis2024solution,minsker2024robust}    & $\surd$   & $\times$& $\times$   & $\times$ \\
g-SLOPE \citep{Groupslope,wei2021sparse,chen2023group,feser2023sparse}   & $\surd$   & $\times$ & $\times$   & $\surd$   \\
Stepdown (FDP) \citep{kFWER2}  & $\times$    &$\surd$   & $\times$   & $\times$       \\
Stepdown ($k$-FWER) \citep{kFWER2}& $\times$    & $\times$  & $\surd$    &$\times$      \\ 
F-SLOPE (Ours)& $\surd$    & $\surd$   & $\times$  & $\times$     \\ 
$k$-SLOPE (Ours)& $\surd$    & $\times$ & $\surd$    &$\times$     \\ 
gF-SLOPE (Ours)& $\surd$    & $\surd$& $\times$    & $\surd$    \\ 
g$k$-SLOPE (Ours)& $\surd$    & $\times$  & $\surd$    & $\surd$     \\ 
\hline
\end{tabular}
}
\label{table1_compare}
\footnotetext[1]{The group-level definitions of these three criteria (gFDR, gFDP and g$k$-FWER) are considered with respect to g-SLOPE, gF-SLOPE and g$k$-SLOPE.}
\end{table*}

%\subsection{Main Contribution}
To fill this gap, we propose new SLOPE and g-SLOPE approaches, namely $k$-SLOPE, F-SLOPE, g$k$-SLOPE, and gF-SLOPE, which enable feature selection at either a single regressor or grouped features with $k$-FWER and FDP control, respectively. In contrast to the previous methods relying on the BH procedure \citep{benjamini-1995bh}, the proposed SLOPEs and g-SLOPEs depend on the stepdown procedure \citep{kFWER2}, which offers greater feasibility and adaptivity \citep{slope,slopestatics, Groupslope} as summarized in Table \ref{table1_compare}. The main contributions of this paper are outlined as follows:
\begin{itemize}
\item \emph{New SLOPEs for the k-FWER and the FDP control}. We integrate the SLOPE \citep{slope} and the stepdown procedure \citep{kFWER2} in a coherent manner for $k$-FWER and FDP control, and formulate the corresponding convex optimization problems. Similar to the flexible knockoffs inference in \citep{zhao-2022error}, our approaches can also avoid the complex calculations for p-value and can be implemented feasibly.

\item \emph{New g-SLOPEs for the g$k$-FWER and gFDP control.} We extend the idea of $k$-SLOPE and F-SLOPE to deal with the circumstance when one aims at selecting the whole subsets of explanatory variables with group structures instead of single regressors \citep{Groupslope}. 

\item \emph{Theoretical guarantees and empirical effectiveness}. 
Under the orthogonal design setting, the $k$-FWER and FDP can be provably controlled at a prespecified level for the proposed SLOPEs and $g$-SLOPEs, respectively. 
Extensive simulation experiments validate the effectiveness and flexibility of the stepdown SLOPEs for $k$-FWER and FDP control, as well as the group stepdown SLOPEs for g$k$-FWER and gFDP control, thereby verifying our theoretical findings. 
\end{itemize}

\section{Related Work}\label{sec2}
To better emphasize the novelty of the proposed approaches, we review the related SLOPE \citep{slope} and Group SLOPE \citep{Groupslope} methods as well as the relationship among FDR, $k$-FWER and FDP criteria.

{\bf{SLOPE} and \bf{Group SLOPE} Methods.}
SLOPE \citep{slope} can be considered as a natural extension of Lasso \citep{lasso}, where the regression coefficients are penalized according to their ranks. 
One notable choice of the regularization sequence $\{\lambda_i\}$ is given by the BH \citep{benjamini-1995bh} critical values $\lambda_{\mathrm{BH}}(i)=\Phi^{-1}(1-\frac{iq}{2m})$, where $q \in (0,1)$ is the desired FDR level, $m$ is the characteristic number and $\Phi(\cdot)$ is the cumulative distribution function of a standard normal distribution. The primary motivation behind SLOPE is to provide finite-sample guarantees for regression estimation and FDR control, where FDR is the expected proportion of irrelevant regressors among all selected predictors. When $X$ is an orthogonal matrix, SLOPE with $\lambda_{\mathrm{BH}}$ controls FDR at the desired level in theory. Additionally, a notable property is that SLOPE does not require knowledge of the sparsity level, yet it automatically yields optimal total squared errors across a wide range of $\ell_0$-sparsity classes. 

To improve computational efficiency, a sparse semismooth Newton-based augmented Lagrangian method was proposed to solve the more general SLOPE model \citep{luo2019solving}. A heuristic screening rule for SLOPE, based on the strict rules for Lasso, was first introduced to improve the computational efficiency of SLOPE, particularly for estimating the complete regularization path \citep{strongslope}. And Larsson et al. also proposed a new fast algorithm to solve the SLOPE optimization problem, which combined proximal gradient descent and proximal coordinate descent steps \citep{larsson}. Besides the above works on algorithm optimization, there are extensive studies on SLOPE with properties \citep{slopestatics,bellec,kos}, model improvements \citep{Groupslope,lee,riccobello,abslope} and applications \citep{brzyski,kremer}. 
Liang et al. \citep{liang2023stepdown} (prior conference version of this article) improved the classical SLOPE by controlling other criteria, such as $k$-FWER and FDP, and demonstrated both empirical performance and theoretical guarantees.
In addition, to select variables where relevant features can be grouped, Brzyski et al. \citep{Groupslope} developed an advanced Group SLOPE ($g$-SLOPE) to achieve controlled group-wise feature selection.
However, to the best of our knowledge, no previous studies have investigated group-wise SLOPE-based feature selection with $k$-FWER or FDP control guarantees.

{\bf{Statistical Metrics: FDR, $k$-FWER and FDP}.}
Benjamini and Hochberg formulated the BH procedure to control the FDP expectation, also known as \emph{FDR control} \citep{benjamini-1995bh}. Subsequently, Lehmann and Romano proposed both the single-step and stepdown procedures to ensure $k$-FWER control \citep{kFWER2}. They also considered the FDP control and provided two stepdown procedures for controlling the FDP under mild conditions, with or without dependence assumptions on the p-value structure. With the help of the stepdown procedures \citep{kFWER2}, several studies have been conducted on feature selection using the $k$-FWER control \citep{stepup,romano2007control,aleman2017effects,zhao-2022error} and the FDP control \citep{stepup,romano2007control,delattre2015new,zhao-2022error}. However, most of these procedures may depend on p-values to assess the importance of each feature or the validity of the underlying structures. Moreover, the traditional calculation of p-values usually relies on large-sample asymptotic theory, which may no longer hold in high-dimensional finite samples \citep{Cands2016PanningFG}. 

Furthermore, it is of great importance to explain the relationship between FDR, FDP, and $k$-FWER. Given $\gamma,\alpha\in (0,1)$, the FDP control means the $\rm{Prob}(\mathrm{FDP}>\gamma)$ at the level $\alpha$. Recall that the FDP concerns \begin{equation}
\rm{Prob}\{\mathrm{FDP}>\gamma\}<\alpha,
\end{equation}
and FDR is the expectation of FDP, i.e., $\mathrm{FDR}=\mathbb{E}\mathrm{(FDP)}$. 
Then following \citep{slope,slopestatics}, we can verify
\begin{equation}
\begin{aligned}
\mathrm{FDR}
\leq \gamma \rm{Prob}\{\mathrm{FDP} \leq \gamma\}+\rm{Prob}\{\mathrm{FDP}>\gamma\},
\end{aligned}
\end{equation}
and
\begin{equation}
\frac{\mathrm{FDR}-\gamma}{1-\gamma} \leq \rm{Prob}\{\mathrm{FDP}>\gamma\} \leq \frac{\mathrm{FDR}}{\gamma},
\end{equation}
where the last inequality follows from Markov’s inequality. If a method controls FDR at level $q$, then it also controls $\mathrm{FDP}\leq q/\gamma$. Conversely, if the FDP is controlled, i.e. $\rm{Prob}(\mathrm{FDP}>\gamma)<\alpha$, and then the FDR is bounded by $(1-\gamma)\alpha+\gamma$. Therefore, a procedure with FDP control can often control the FDR \citep{vanderLaanDudoitPollard}. Moreover, Farcomeni pointed out that, compared with the FDR control, the $k$-FWER control can be more desirable where powerful selection results can be obtained \citep{Farcomen} in settings where more powerful selection results are required \citep{Farcomen}.
\section{Preliminaries}
\label{sec3}
This section recalls some necessary backgrounds involved in this paper, e.g., the criterion definitions (for FDP, FDR and $k$-FWER), classical sparse regressors (including Lasso \citep{lasso}, Group Lasso \citep{group}, SLOPE \citep{slope}), Group SLOPE \citep{Groupslope} and the stepdown procedure \citep{kFWER2}.

\subsection{Problem Formulation}
Let $\mathcal{X}\subset\mathbb{R}^m$ and $\mathcal{Y}\subset \mathbb{R}$ be the compact input space and corresponding output space, respectively. Considering samples $\{(x_i,y_i)\}_{i=1}^n$ are independently drawn from an unknown distribution on $\mathcal{X} \times \mathcal{Y}$. Denote 
\begin{equation}
X:=(X_1,X_2,\cdots,X_n)^T\subset\mathbb{R}^{n\times m}
\end{equation}
and 
\begin{equation}
y=(y_1,y_2,\cdots,y_n)^T\in\mathbb{R}^n.
\end{equation}

The module length of each column vector of $X$ is equal to 1. The following multiple linear regression model generates the output vector $y$:
\begin{align}
y = X\beta+\epsilon,\label{linear}
\end{align}
where $\beta\in\mathbb{R}^m$ represents the coefficient vector and $\epsilon \sim N(0,\sigma^2I_n)$. 

In sparse high-dimensional regression, we usually assume that $\beta$ satisfies a sparse structure. Let $V$ be the number of false selected features and let $R$ be the total number of identified features. The FDP, FDR, and $k$-FWER are respectively defined as 
\begin{equation}
\begin{aligned}
&\mathrm{FDP}=\frac{V}{\max\{R,1\}},~~ \mathrm{FDR}=\mathbb{E}(\mathrm{FDP}),
\end{aligned}
\end{equation}
and
\begin{equation}
\begin{aligned}
&k\text{-}\mathrm{FWER}={\rm{Prob}}\{V\geq k\}. 
\end{aligned}
\end{equation}

\subsection{Lasso and SLOPE}
Tibshirani et al. developed the Lasso model \citep{lasso} to remove irrelevant, redundant, or strongly correlated features from high-dimensional data, thereby selecting the remaining important features for model construction. This approach plays a vital role in machine learning and statistical modeling. The mathematical learning object of Lasso regression is defined as follows,
\begin{align}
\mathop{\arg\min}_{\beta\in \mathbb{R}^m} \left\{\frac{1}{2}||y-X\beta||^2+\lambda_L \|\beta\|_1\right\},
\label{eq_lasso}
\end{align}
where $\beta$ is the regression coefficient. The $\lambda_L$ is the fixed penalty coefficient of the sparse $l_1$ norm, which represents the trade-off between model fitting and the sparsity effect. The proper setting of $\lambda_L$ is a challenging task, as Bogdan et al. \citep{slope} have discovered. Specifically, when using the Cross-Validation strategy to find a suitable $\lambda_L$, Lasso performs poorly at feature selection, leading to relatively high FDR in high-dimensional data.

By replacing the $\ell_1$ penalty in \eqref{eq_lasso} with the sorted $\ell_1$ penalty, Bogdan et al. proposed the SLOPE method \citep{slope} for controlled feature selection in the high-dimensional sparse cases. The learning scheme of SLOPE is formulated as 
\begin{equation}
\mathop{\arg\min}_{\beta\in \mathbb{R}^m} \left\{\frac{1}{2}||y-X\beta||^2+J_\lambda(\beta)\right\},\label{SLOPEeq}
\end{equation}
where $J_\lambda(\beta)=\sum_{i=1}^m\lambda_i|\beta|_{(i)}$ is the sorted $l_1$ norm, the regularization parameters $\lambda_{1}\geq\cdots\geq\lambda_{m}\geq0$ and the regression coefficients $|\beta|_{(1)}\geq |\beta|_{(2)}\geq\cdots\geq |\beta|_{(m)}$ are all non-negative and non-decreasing sequences. When $\lambda_1=\lambda_2=\cdots=\lambda_m$, the optimizing scheme (\ref{SLOPEeq}) obviously reduces to the Lasso \eqref{eq_lasso}. 

Given a desired level $q$, SLOPE controls FDR by exploiting the sequence of parameters 
\begin{equation}
\mathrm{\lambda_{BH}=\{\lambda_{BH}(1),\lambda_{BH}(2),\cdots,\lambda_{BH}}(m)\},
\end{equation}
with
\begin{equation}
\mathrm{\lambda_{BH}}(i)=\sigma\cdot\Phi^{-1}(1-\frac{i q}{2m})\label{lambdaBH},
\end{equation}
where $\Phi(\cdot)$ denotes the cumulative distribution function of the standard normal distribution under orthogonal design. 

\begin{lemma} \label{theorem1}
\citep{slope} In the linear model with the orthogonal design $X$ and $\epsilon\sim N (0,\sigma^2 I_n)$, the SLOPE \eqref{SLOPEeq} with the regularization parameter sequence (\ref{lambdaBH}) satisfies 
$\mathrm{FDR} \leq \frac{m_{0} q}{m}$,
where $m_0$ is the number of true null hypotheses and $q$ is the desired FDR level.
\end{lemma}

\begin{algorithm}[!t]
\caption{Accelerated proximal gradient algorithm for SLOPE (\ref{SLOPEeq})}
\textbf{Input}: Training set $X\in\mathbb{R}^{n\times m}$ and $y\in\mathbb{R}^n$ and parameter $\lambda=(\lambda_1,\lambda_2,...,\lambda_m)$.\\
\textbf{Initialization}: $a^0\in\mathbb{R}^m$, $b^0=a^0$ and $\theta_0=1$.
% \textbf{Initialization}: $a^0\in\mathbb{R}^m$

\begin{algorithmic}[0] %[1] enables line numbers
\FOR{k = 0,1,$\cdots$}
\STATE $b^{k+1}=\operatorname{prox}_{t_{k} J_{\lambda}}\left(a^{k}-t_{k} X^{T}\left(X a^{k}-y\right)\right)$
\STATE $\theta_{k+1}^{-1}=\frac{1}{2}\left(1+\sqrt{\left.1+4 / \theta_{k}^{2}\right)}\right.$
\STATE $a^{k+1}=b^{k+1}+\theta_{k+1}\left(\theta_{k}^{-1}-1\right)\left(b^{k+1}-b^{k}\right)$
\ENDFOR
\end{algorithmic}
\textbf{Output}: $a$ satisfying the stopping criteria.
\label{accslope}
\end{algorithm}

Lemma \ref{theorem1} illustrates the theoretical guarantee of FDR control for SLOPE equipped with $\lambda_{\mathrm{BH}}$ induced by the BH procedure \citep{benjamini-1995bh}. In this paper, the proposed new SLOPEs and $g$SLOPEs are not limited to the FDR control, but extended to the $k$-FWER and FDP control by replacing the BH procedure with the stepdown procedure \citep{kFWER2}.

From the perspective of computation feasibility, the optimization objective function of SLOPE (\ref{SLOPEeq}) is convex but non-smooth, which can be implemented efficiently by the proximal gradient descent algorithm \citep{slope}. For completeness, we present the computing steps of SLOPE in Algorithm \ref{accslope}, which also applies to our variants of SLOPE. Here, $J_{\lambda}=\sum_{i=1}^m\lambda_i|\beta|_{(i)}$ and the step sizes are selected by backtracking line search and satisfy $t_k<2/||X||^2$ \citep{beck2009fast,becker2011templates}. Moreover, Bogdan et al. also derived concrete stopping criteria through duality theory \citep{slope}.

\subsection{Group Lasso and Group SLOPE}

In real-world applications, features in high-dimensional data often operate in groups, and merely considering single regressors fails to yield effective feature selection. To address such issues, Group Lasso partitions the feature set into groups and employs group-sparse regularization \citep{yuan2006model,simon2013sparse,sglasso-NIPS19}.

Similar to the concept of Group Lasso, the Group SLOPE (g-SLOPE) \citep{Groupslope} extends the classic SLOPE \citep{slope} to solve the case where one aims at selecting whole groups of explanatory variables instead of single regressors. Such groups can be formed by clustering strongly correlated predictors or groups of dummy variables corresponding to different levels of the same qualitative predictor.

Assume $I=\{I_1\cup I_2...\cup I_t\}$ is some partition of the set $\{1,...,n\}$, i.e. $I_i$s are nonempty sets, $I_i\cap I_j=\emptyset $ for $i\neq j$. 
The data generation mechanism with $t$ groups for linear regression is
\begin{equation}
y=\sum_{i=1}^t X_{I_i}\beta_{I_i}+\epsilon,
\label{groupform}
\end{equation}
where $X_{I_i}$ is the submatrix of $X$ composed of columns indexed by $I_i$ and $\beta_{I_i}$ is the restriction of $\beta$ to indices from the set $I_i$. In addition, $X$ can be any matrix, especially any linear relationship within the submatrices $X_{I_i}$ is allowed. 
Assume that $\epsilon \sim N(0,\sigma^2I_n)$ where $\sigma$ represents the standard deviation.
Let the notations $l_1,...,l_t$ refer to the ranks of submatrices $X_{I_1},...,X_{I_t}$.

The frequently used formulation of Group Lasso is defined as
\begin{align}
\mathop{\arg\min}_{\beta\in \mathbb{R}^m}\left\{\frac{1}{2}||y-X\beta||^2+\sigma \lambda_{GL} \sum_{i=1}^t \sqrt{\left|I_i\right|}\left\|\beta_{I_i}\right\|_2\right\},
\end{align}
where $\lambda_{GL}$ is the fixed coefficient of the Group Lasso penalty. $\left|I_i\right|$ denotes the number of elements in set $I_i$, $X_{I_i}$ is the submatrix of $X$ composed of columns indexed by $I_i$ and $\beta_{I_i}$ is the restriction of regression coefficient $\beta$ to indices from $I_i$. The limitations of Group Lasso are similar to those of Lasso, including the high FDR in high-dimensional data. Thus, it's of great importance to integrate the concept of SLOPE and grouped feature selection.

Analogous to the Group Lasso definition and the concept of SLOPE. Define the sorted $\ell_1$-induced penalty $J_\lambda$ as $J_\lambda(v)=\sum_{i=1}^t \lambda_i v_{(i)}, \text{~where~} v_{(1)} \geq v_{(2)} \geq \cdots \geq v_{(t)}$. Then the form of the g-SLOPE \citep{Groupslope} could be formulated as 
\begin{equation}
\beta^{gS} \in \mathop{\arg\min}_{\beta\in \mathbb{R}^m}\frac{1}{2}||y-X\beta||^2+\sigma J_\lambda (W\| \beta \|_{X,I}),\label{GroupSLOPEeq}
\end{equation}
where $\lambda_1, \cdots, \lambda_m$ is a given non-increasing sequence of non-negative tuning parameters. The sparsity-induced penalty $J_\lambda(\beta)=\sum_{i=1}^m\lambda_i|\beta|_{(i)}$ has been defined in Eq.\eqref{SLOPEeq}. $W$ is a diagonal matrix with $W_{i,i}=w_{i}$ and $0$ otherwise, where $w_i$ is the group-wise weight and practically set according to the root of group size $w_i=\sqrt{|I_i|}$ \citep{Groupslope} or $w_i=1/\sqrt{|I_i|}$. 

The defined group-wise regularization $\| \beta \|_{X,I}$ implies
\begin{equation}
\|\beta\|_{X, I}=\left(\left\|X_{I_1} \beta_{I_1}\right\|_2,\left\|X_{I_2} \beta_{I_2}\right\|_2, \ldots,\left\|X_{I_t} \beta_{I_t}\right\|_2\right)^T. \label{GS_group}
\end{equation}

Then FDR no longer applies because the problem being solved has changed. So Brzyski et al. define the group false discovery rate (gFDR) \citep{Groupslope}. 
Let the number of all groups selected by g-SLOPE be $Rg$, and the number of groups falsely discovered by g-SLOPE be $Vg$. Thus we have
\begin{equation}
Rg:=|\{i:||X_{I_i}\beta_{I_i}^{gS}||_2 \neq 0\},
\end{equation}
where $\beta_{I_i}^{gS}$ at i-th group has been defined in Eq.\eqref{GroupSLOPEeq} and \eqref{GS_group}, and there also holds
\begin{equation}
Vg:=|\{i:||X_{I_i}\beta_{I_i}||_2=0, ||X_{I_i}\beta_{I_i}^{gS}||_2\}|\neq 0 \},
\end{equation}
where the estimation of $\| \beta \|_{X,I}$ is defined by the indices corresponding to nonzeros of $\| \beta^{gS}\|_{X,I}$. 
The gFDR is defined as 
\begin{equation}
\mathrm{gFDR}:=\mathbb{E}(\frac{Vg}{max\{Rg,1\}}).
\end{equation}

The goal is to identify the regularized sequences for g-SLOPE so that gFDR can be controlled at any level $\alpha \in(0,1)$. The following is an expansion of the regularization parameter form for the orthogonal and general cases. When the design matrix $X$ is orthogonal matrix, the sequence of regularizing parameters $\lambda_{max}=(\lambda_1^{max},\lambda_2^{max},...,\lambda_m^{max})^{T}$ is define as 
\begin{equation}
\lambda_i^{max}:=\max_{j=1,...,m}{\frac{1}{w_j}F_{\mathcal{X}_{l_j}}^{-1}(1-\frac{qi}{m})},
\end{equation}
where $F_{\mathcal{X}_{l_j}}$ is a cumulative distribution function of $\mathcal{X}$ distribution with $l_j$ degrees of freedom. Then it holds that
\begin{equation}
\mathrm{gFDR}\leq \alpha\cdot\frac{m_0}{m}.
\end{equation}

Moreover, \citep{Groupslope} proved that the resulting procedure adapts to unknown sparsity and is asymptotically minimax for estimating the proportions of variance in the response variable explained by regressors from different groups.

\subsection{Stepdown Procedure}
The stepdown procedure \citep{kFWER2} aims to control $k$-FWER and FDP, i.e., given the thresholds $\alpha,\gamma\in (0,1)$, we aim to control
\begin{align}
&k\text{-}\mathrm{FWER}\leq\alpha
\label{kFWERcantrol}
\end{align}
and
\begin{align}
&\rm{Prob}\{\mathrm{FDP}>\gamma\}\leq\alpha.
\label{FDPcontrol}
\end{align}

Suppose that there are $m$ individual tests $H_{1},...,H_m$, whose corresponding p-values are $\hat{p}_1,...,\hat{p}_m$. Let $\hat{p}_{(1)}\leq \hat{p}_{(2)}\leq...\leq \hat{p}_{(m)}$ be the ordered p-values and let the non-negative and non-decreasing sequence $\alpha_1\leq\alpha_2...\leq \alpha_m$ be the $k$-FWER thresholds. The hypotheses corresponding to the sorted p-values are defined as $H_{(1)},H_{(2)}...,H_{(m)}$. Then the stepdown procedure can be defined step by step as follows:

\textbf{Step\,0}: Let $i=0$.

\textbf{Step\,1}: If $\hat{p}_{(i+1)}\geq \alpha_{i+1}$, go to \textbf{Step 2}. Otherwise, set $i=i+1$ and repeat \textbf{Step 1}.

\textbf{Step\,2}: Reject $H_{(j)}$ for $j\leq k$ and accept $H_{(j)}$ for $j>k$.

In other words, if $p_{(1)}>\alpha_1$, no null hypotheses are rejected. Otherwise, if $H_{(1)},H_{(2)}...,H_{(r)}$ are rejected, the largest $r$ satisfies
\begin{equation}
p_{(1)}\leq \alpha_1,p_{(2)}\leq\alpha_2,...,p_{(r)}\leq\alpha_r.\label{down}
\end{equation} 

Based on the stepdown procedure, Lehmann and Romano provided two different thresholds to ensure the $k$-FWER control and the FDP control, respectively \citep{kFWER2}.
\begin{lemma}\citep{kFWER2}\label{theorem2}
For testing $H_i,i= 1,...,m$, given $k$ and $\alpha\in(0,1)$, the stepdown procedure described in (\ref{down}) with 
\begin{equation}
\alpha_{i}= \begin{cases}\frac{k \alpha}{m}, & i \leq k \\ \frac{k \alpha}{m+k-i}, & i>k\end{cases}	\label{kFWER_threlod}
\end{equation}
controls the $k$-FWER, then (\ref{kFWERcantrol}) holds.

\end{lemma}

\begin{lemma}\citep{kFWER2}\label{theorem3}
For testing $H_i,i= 1,...,m$, given $\alpha,\gamma\in (0,1)$, if the p-values of false null hypotheses are independent of the true ones, the stepdown procedure described in (\ref{down}) with 
\begin{equation}
\alpha_{i}=\frac{(\lfloor\gamma i\rfloor+1) \alpha}{m+\lfloor\gamma i\rfloor+1-i}\label{FDP_threlod}
\end{equation}
controls the FDP in the sense of (\ref{FDPcontrol}).
\end{lemma}

Lemmas \ref{theorem2} and \ref{theorem3} demonstrate that the stepdown procedure enjoys theoretical guarantees for $k$-FWER control and FDP control under ingenious selections of $\alpha_i$. Indeed, these theoretical properties of the stepdown procedure motivate our designs for new SLOPE and g-SLOPE algorithms. 

\section{Methodology}\label{sec4}
This section integrates the stepdown procedure \citep{kFWER2} into the classical SLOPE \citep{slope} and g-SLOPE \citep{Groupslope} to formulate our new stepdown-based approaches for controlled feature selection, ensuring both $k$-FWER control and FDP control. First, we provide the tuning parameter sequences for the $k$-FWER and FDP controls under the orthogonal design. Furthermore, we present an intuitive theoretical analysis for selecting regularization parameters under general settings.

\subsection{Orthogonal Design}
It has been illustrated in \citep{slope} that there is a link between multiple tests and model selection for SLOPE under orthogonal design. Following this line, we assume that $X$ is an $n\times m$ dimensional orthogonal matrix, i.e, $X^{T}X=I_m$ and $\epsilon \sim N(0,\sigma^2I_n)$ is an $n$-dimensional column vector with known variance. Then, the linear regression model $y=X\beta+\epsilon$
is further transformed into
\begin{equation}
\tilde{y}=X^{T}y=\beta+X^{T}\epsilon\sim N(\beta,\sigma^2I_p).
\end{equation}

It is well known that selecting effective features can be reduced to a multiple-hypothesis testing problem. Denote $m$ hypotheses as $H_i:\beta_i=0,1\leq i\leq m$. If $H_i$ is rejected, $\beta_i$ is considered as an informative feature and vice versa. Bogdan et al. provided the selection mechanism for regularization parameters of SLOPE \citep{slope} via the BH procedure \citep{benjamini-1995bh} under an  For brevity, we refer to the proposed methods as $k$-SLOPE and F-SLOPE, respectively, with respect to FWER and FDP. 

The regularization scheme of $k$-SLOPE is formulated as
\begin{equation}
\mathop{\arg \min}_{\beta\in \mathbb{R}^m}\frac{1}{2}\|y-X \beta\|_{l_2}^{2}+ \sigma\cdot\sum_{i=1}^m\lambda_{k\text{-}\mathrm{FWER}}(i)|\beta|_{(i)},\label{kFWER_model}
\end{equation} 
where
\begin{equation} \label{sequence1}
\lambda_{k\text{-}\mathrm{FWER}}(i)= \begin{cases}\Phi^{-1}(1-k \alpha/{2m}), & i \leq k \\ \Phi^{-1}(1-k\alpha/{2(m+k-i)}), & i>k.\end{cases}
\end{equation}

To characterize the theoretical property of the proposed $k$-SLOPE, we introduce the following lemmas stated in \citep{slope,topk}.

\begin{lemma}\citep{slope}
Let $H_i$ be a null hypothesis and let $r\geq 1$. Then we can derive that
\begin{equation}\label{lemma1}
\begin{aligned}
&\left\{y: H_{i} \text { is rejected and } R =r\right\}
=\left\{y:\left|y_{i}\right|>\lambda_{r} \text { and } R=r\right\}.
\end{aligned}
\end{equation}
\end{lemma}

\begin{remark}
Lemma 1 demonstrates that when the number of selected features $R$ is equal to $r$ ($r\geq1$), the condition for $H_i$ to be rejected is that its corresponding $|y_i|$ is greater than $\lambda_r$.
\end{remark}

\begin{lemma}\citep{slope}\label{lemma2}
Consider applying the SLOPE procedure to $\tilde{y}=(y_1,...,y_{i-1},y_{i+1},...y_p)$ with weight $\tilde{\lambda}=(\lambda_2,...,\lambda_p)$ and let $\tilde{R}$ be the number of rejections among this procedure. Then with $r>1$, we have
\begin{equation}
\{y:|y_i|>\lambda_r\text{ and } R=r\}\subset\{y:|y_i|>\lambda_r\text{ and }\tilde{R}=r-1\}.
\end{equation}
\end{lemma}

\begin{remark}
Lemma 2 states that if $y_i$ is removed at random, the number of selected features is greater than that in the previous set.
\end{remark}

\begin{lemma}\citep{topk}\label{lemma3}
Let $s_{[k]}$ be the top-k element of a set $S=\{s_1,...,s_n\}$, such as $s_{[1]}\geq s_{[2]}\geq...\geq s_{[n]}$. $\sum_{i=1}^k s_{[i]}$ is a convex function of $(s_1,...,s_n)$. Furthermore, for $s_i\geq0$ and $i=1,...,n$, we have $\sum_{i=1}^k s_{[i]}=\min_{\lambda\geq0}\{k\lambda+\sum_{i=1}^n [s_i-\lambda]_+\}$, where $[a]_+=\max\{0,a\}$ is the hinge function.
\end{lemma}

\begin{remark}
Lemma 3 gives the equivalent form of the sum of the first $k$ non-negative values. 
\end{remark}

The $k$-SLOPE equipped with \eqref{sequence1} yields the following theoretical property.
\begin{theorem} \label{our_theorem1}
In the linear model (\ref{linear}) with the orthogonal matrix $X$ and noise $\epsilon\sim N(0,\sigma^2I_n)$, given $k$ and $\alpha\in(0,1)$, the $k$-FWER of the $k$-SLOPE model (\ref{kFWER_model}) satisfies (\ref{kFWERcantrol}).
\end{theorem}
\begin{proof}
Suppose that $X$ is the orthogonal design matrix and $\sigma^2 = 1$. Now we have $\tilde{y}=X^{T}y\sim N(\beta,I_m)$ and $m$ null hypotheses $H_i:\beta_i=0, i=1,...,m$, in which the first $m_0$ hypotheses are true null hypotheses. Then order the absolute value of $y$ corresponding to the first $m_0$ true null hypotheses, denoted as 
\begin{equation}
|\hat{q}|_{(1)}\geq|\hat{q}|_{(2)}\geq...\geq|\hat{q}|_{(m_0)}.
\end{equation}

Let $k$ be the largest number of false selected features within the scope of tolerability or the number of indices that satisfy $\beta_i\neq 0$ within $\{1,...,m_0\}$. Assume $m_0\geq k$ or there is nothing left to prove. Then we can derive,
\begin{equation}
\begin{aligned}
k\text{-}\mathrm{FWER}=\mathrm{Prob}(V\geq k)
=\sum_{r=1}^m \mathrm{Prob}(V\geq k \text{ and } R=r).\label{eq1}
\end{aligned}
\end{equation}

Through the stepdown procedure \citep{kFWER2} and Lemma \ref{lemma1} \citep{slope}, $k$-SLOPE commits at least $k$ false rejections if and only if 
\begin{equation}
|\hat{q}|_{(1)}\geq \lambda_r,|\hat{q}|_{(2)}\geq\lambda_r,...,|\hat{q}|_{(k)}\geq\lambda_r,
\end{equation}
when the number of selected features $R = r$. Then we have
\begin{equation}
\begin{aligned}
\mathrm{Prob}(V\geq k \text{ and } R=r)
=&\mathrm{Prob}(|\hat{q}|_{(1)}\geq\lambda_r,...,|\hat{q}|_{(k)}\geq\lambda_r \text{ and } R=r)\\
\leq &\mathrm{Prob}(|\hat{q}|_{(k)}\geq\lambda_r \text{ and } R=r).
\end{aligned}
\end{equation}

Due to Lemma 2 \citep{slope} and the independence between $|\hat{q}|_{(k)}$ and $\tilde{y}$, we have
\begin{align}
&\mathrm{Prob}(V\geq k \text{ and } R=r)
\leq \mathrm{Prob}(|\hat{q}|_{(k)}\geq\lambda_r)\mathrm{Prob}(\tilde{R}=r-1)\label{eq2}
\end{align}

It is not difficult to find $\mathrm{Prob}(|\hat{q}|_{(i)}\geq\lambda_r)$ non-increasing, thus
\begin{align}
\mathrm{Prob}(|\hat{q}|_{(k)}\geq\lambda_r) \leq \frac{1}{k}\sum_{i=1}^k\mathrm{Prob}(|\hat{q}|_{(i)}\geq\lambda_r)\label{eq3}
\end{align}

Next, combined with Lemma \ref{lemma3} \citep{topk}, we can derive the following
\begin{equation}
\begin{aligned}
\frac{1}{k}\sum_{i=1}^k\mathrm{Prob}(|\hat{q}|_{(i)}\geq\lambda_r)
=&\frac{1}{k}\min_{t\geq0}\{kt+\sum_{i=1}^{m_0}[\mathrm{Prob}(|\hat{q}_i|\geq\lambda_r)-t]_+\}\\
\leq&\min_{t_0\leq t\leq \alpha}\{t+\frac{m_0}{k}[t_0-t]_+\},
\end{aligned}
\end{equation}
where $t_0=\mathrm{Prob}(|\hat{q}_i|\geq\lambda_r)=k\alpha/(p+k-r)$. 

Plugging above inequalities into (\ref{eq1}) gives
\begin{equation}
\begin{aligned}
k\text{-}\mathrm{FWER} 
= &\sum_{r=1}^m\mathrm{Prob}(|\hat{q}|_{(1)}\geq\lambda_r,...,|\hat{q}|_{(k)}\geq\lambda\text{ and } R=r)\\
\leq&\sum_{r=1}^m\min_{t_0\leq t\leq \alpha}\{t+\frac{p_0}{k}[t_0-t]_+\}\mathrm{Prob}(\tilde{R}=r-1)\\
\leq &\sum_{r=1}^m \frac{k\alpha}{p+k-r} \mathrm{Prob}(\tilde{R}=r-1)\leq \alpha.
\end{aligned}
\end{equation}

This completes the proof. 
\end{proof}

\begin{remark}
Theorem \ref{our_theorem1} illustrates that $k$-SLOPE controls the $k$-FWER under the orthogonal design. Although the $\lambda_{k\text{-}\mathrm{FWER}}(i)$'s are chosen with reference \citep{kFWER2}, (\ref{kFWER_model}) is not equivalent to the stepdown procedure described above. We also empirically validate this theoretical finding under different experiment settings in Section \ref{sec5}.
\end{remark}

Generally, the number of falsely selected features that people are willing to accept is directly proportional to the number of identified features. Therefore, we may no longer be concerned about $k$-FWER, but about FDP. Analogous to (\ref{kFWER_model}), the convex optimization problem of F-SLOPE is formulated as
\begin{equation}
\mathop{\arg \min}_{\beta\in \mathbb{R}^m}\frac{1}{2}\|y-X \beta\|_{l_2}^{2}+ \sigma\cdot\sum_{i=1}^m\lambda_{\mathrm{FDP}}(i)|\beta|_{(i)},\label{FDP_model}
\end{equation}
where 
\begin{equation}
\begin{aligned}
&\lambda_{\mathrm{FDP}}(i)=\Phi^{-1}(1-\frac{(\lfloor\gamma i\rfloor+1) \alpha}{2(m+\lfloor\gamma i\rfloor+1-i)}).
\end{aligned}
\end{equation}

Furthermore, the proposed selection algorithm of regularization parameters also produces the following theoretical guarantee. 

\begin{theorem}\label{our_theorem2}
In the linear model (\ref{linear}) with the orthogonal matrix $X$ and noise $\epsilon\sim N(0,\sigma^2I_n)$, given $\alpha,\gamma\in (0,1)$, the FDP of the F-SLOPE model (\ref{FDP_model}) satisfies (\ref{FDPcontrol}).
\end{theorem}
\begin{proof}
Let the number of true hypotheses be non-zero, i.e., $m_0>0$; otherwise, it is not necessary to prove it again. Given $\gamma\in(0,1)$, under the same assumptions of Theorem 4, we have the following
\begin{equation}\label{eq4}
\begin{aligned}
\mathrm{Prob}(\mathrm{FDP}>\gamma)
=\sum_{r=1}^m\mathrm{Prob}(\frac{V}{R}>\gamma\text{ and } R=r)
=\sum_{r=1}^m\mathrm{Prob}(V\geq k(R) \text{ and }R=r),
\end{aligned}
\end{equation}
where $k(R)=\lfloor{\gamma R}\rfloor+1$. We observed that the FDP control in (\ref{eq4}) is similar to the $k$-FWER in (\ref{eq1}) and the only difference is whether the value of $k$ is affected by the number of selected features. Based on Lemma \ref{lemma1} \citep{slope} and the stepdown procedure \citep{kFWER2}, we can derive 
\begin{equation}
\begin{aligned}
\mathrm{Prob}(V\geq k(R) \text{ and } R=r)
=&\mathrm{Prob}(|\hat{q}|_{(1)}\geq\lambda_r,...,|\hat{q}|_{(k(R))}\geq\lambda_r\text{ and }R=r)\\
\leq &\mathrm{Prob}(|\hat{q}|_{(k)}\geq\lambda_r \text{ and } R=r).
\end{aligned}
\end{equation}

Analogy to (\ref{eq2}) and (\ref{eq3}), we have
\begin{equation}
\begin{aligned}
\mathrm{Prob}(V\geq k(R) \text{ and } R=r)
\leq&\mathrm{Prob}(|\hat{q}|_{(k(R))}\geq\lambda_r) \cdot \mathrm{Prob}(\tilde{R}=r)\\
\leq&\frac{1}{k(R)}\sum_{i=1}^{k(R)}\mathrm{Prob}(|\hat{q}|_{(i)}\geq\lambda_r) \cdot \mathrm{Prob}(\tilde{R}=r-1).
\end{aligned}
\end{equation}

Then combined with Lemma \ref{lemma3} \citep{topk}, we have 
\begin{equation}
\begin{aligned}
\frac{1}{k(R)}\sum_{i=1}^{k(R)}\mathrm{Prob}(|\hat{q}|_{(i)}\geq\lambda_r)
=&\frac{1}{k(R)}\min_{t\geq0}\{k(R)\cdot t+\sum_{i=1}^{m_0}[\mathrm{Prob}(|\hat{q}_i|\geq\lambda_r)-t]_+\}\\
\leq&\min_{t_0\leq t\leq \alpha}\{t+\frac{p_0}{k(R)}[t_0-t]_+\},
\end{aligned}
\end{equation}
where $t_0=(k(r)+1)q/(p+k(r)+1-r)$. 

Plugging above inequalities into (\ref{eq4}) gives
\begin{equation}
\begin{aligned}
\mathrm{Prob}(\mathrm{FDP}>\gamma)
=&\sum_{r=1}^m{P}(V\geq k(R)\text{ and } R=r)\\
\leq&\sum_{r=1}^m \min_{t_0\leq t\leq q}\{t+\frac{m_0}{k(R)}[t_0-t]_+\}\mathrm{Prob}(\tilde{R}=r-1)\\
\leq &\sum_{r=1}^m\frac{k(r)q}{p+k(r)-r}\mathrm{Prob}(\tilde{R}=r-1)\leq \alpha,
\end{aligned}
\end{equation}
which completes the proof. 
\end{proof}
\begin{remark}
Theorem \ref{our_theorem1} assures the ability of FDP control for F-SLOPE under the orthogonal experiments. The latter orthogonal experiments also support the conclusion (\ref{FDP_model}), and the $k$-SLOPE model (\ref{kFWER_model}) serves as the selection mechanism for the sequence of penalty parameters. The latter orthogonal experiments also support the conclusion. Moreover, the optimization algorithms for $k$-SLOPE and F-SLOPE are the same as that for SLOPE, since they are all convex and non-smooth. More optimization details are present in Algorithm \ref{accslope}.
\end{remark}

\subsection{General Setting}
Typically, SLOPE is challenging to establish solid theoretical guarantees for the FDR control under a non-orthogonal setting \citep{slope}. Hence, $k$-SLOPE and F-SLOPE may also face degraded performance under such a general setting. Fortunately, Bogdan et al. utilized their qualitative insights to make an intuitive adjustment to the regularization parameter sequence and demonstrated its empirical effectiveness \citep{slope}. Analogously to SLOPE, we present the regularization-parameter forms of $k$-SLOPE and F-SLOPE via theoretical analysis in a general setting.

Assume that $k$-SLOPE and F-SLOPE correctly detect these features and estimate the signs of the regression coefficients accurately. Let $X_S$ and $\beta_S$ be the subset of variables associated with $\beta_i\neq 0$ and the value of their coefficients, respectively. The estimator with nonzero components is  approximated by
\begin{equation}
\hat{\beta}_S\approx(X_S^{T} X_S)^{-1}(X_S^{T}y-\lambda_S)=\hat{\beta}_{\mathrm{LSE}}-(X_S^{T} 
X_S)^{-1}\lambda_S,\label{1}
\end{equation}
where $\lambda_S=(\lambda_1,...,\lambda_{|S|})^{T}$ and $\hat{\beta}_{\mathrm{LSE}}$ is the least-squares estimator of $\beta_S$. 

Inspired by \citep{slope}, we calculate the distribution of $X_i^{T} X_S(\beta_S-\hat{\beta_S})$ to determine the specific forms of the regularization parameters for $k$-SLOPE and F-SLOPE. In light of Eq.\eqref{1}, we can derive
\begin{equation}
\mathbb{E}(\beta_S-\hat{\beta}_S)\approx (X_S^{T} X_S)^{-1}\lambda_S,
\end{equation}
and
\begin{equation}
\mathbb{E}X_i^{T} X_S(\beta_S-\hat{\beta}_S)\approx\mathbb{E}X_i^{T} X_S(X_S^{T} X_S)^{-1}\lambda_S.
\end{equation}

Under the Gaussian design where each element of $X$ is i.i.d drawn from the distribution of $\mathcal{N}(0,1/n)$, we have the following
\begin{equation}
\mathbb{E}(X_i^{T} X_S(X^{T}_S X_S)^{-1}\lambda_S)^2
=\frac{1}{n}\lambda_S^{T}\mathbb{E}(X_S^{T}X_S)^{-1}\lambda_S
=w(|S|)\cdot||\lambda_S||^2
\end{equation}
and
\begin{equation}\label{eq_42}
w(|S|)=\frac{1}{n-|S|-1},
\end{equation}
where $|S|$ is the number of elements of $S$, $i\notin S$ and \eqref{eq_42} relies on the fact that the expected of an inverse $|S|\times|S|$ Wishart matrix with $n$ degrees of freedom is equal to $I_{|S|}/(n-|S|-1)$ \citep{wishart}. 

The $k$-SLOPE begins with $\lambda_{k\mathrm{G}}=\lambda_{k\text{-}\mathrm{FWER}}(1)$. Considering the slight increase in variance, we have
\begin{equation}
\lambda_{k\mathrm{G}}(2)=\lambda_{k\text{-}\mathrm{FWER}}(2)\sqrt{1+w(2)\lambda_{k\mathrm{G}}(1)^2}.
\end{equation}

Thus, the sequence of $\lambda_{k\mathrm{G}}$ can be expressed as
\begin{equation}
\lambda_{k\mathrm{G}}(i)=\lambda_{k\text{-}\mathrm{FWER}}(i)\sqrt{1+w(i-1)\sum_{j< i}\lambda_{k\mathrm{G}}(i)^2}.\label{kgeneral}
\end{equation}

The only difference between F-SLOPE and the $k$-SLOPE is the selection of the coefficient sequence of the penalty term. Similar to (\ref{kgeneral}), F-SLOPE starts with $\lambda_{\mathrm{FG}}=\lambda_{\mathrm{FDP}}(1)$, and then
\begin{equation}
\lambda_{\mathrm{FG}}(i)=\lambda_{\mathrm{FDP}}(i)\sqrt{1+w(i-1)\sum_{j< i}\lambda_{\mathrm{FG}}(i)^2}.\label{fgeneral}
\end{equation}

If the penalty term's coefficient sequence is incremental, $k$-SLOPE and F-SLOPE are no longer convex optimization problems. Denote $k^*:=k(n,m,\alpha)$ as the subscript of global minimum, $k$-SLOPE and F-SLOPE respectively work with
\begin{equation}
\lambda_{\mathrm{kG}^{\star}}(i)= \begin{cases}\lambda_{\mathrm{kG}}(i), & i \leq k^{\star}, \\ \lambda_{k G}\left(k^{\star}\right), & i>k^{\star},\end{cases}
\end{equation}
with $\lambda_{k\mathrm{G}} (i)$ given in (\ref{kgeneral}) and 
\begin{equation}
\lambda_{\mathrm{FG}}(i)= \begin{cases}\lambda_{\mathrm{FG}}(i), & i \leq k^{\star}, \\ \lambda_{F G}\left(k^{\star}\right), & i>k^{\star},\end{cases}
\end{equation}
with $\lambda_{\mathrm{FG}}(i)$ defined in (\ref{fgeneral}). %We empirically verify the validity of these regularization parameters.
When the design matrix isn't Gaussian or its columns aren't independent, we can employ the Monte Carlo estimate of the correction \citep{hammersley1954poor} instead of $w(i-1)\sum_{j<i}\lambda(i)^2$ in the formulas (\ref{kgeneral}) and (\ref{fgeneral}). 

\subsection{Group Stepdown SLOPE}
Inspired by Group SLOPE \citep{Groupslope}, we further extend the Stepdown SLOPE idea to select entire subsets of informative features with group structures, rather than single regressors. We formulate the corresponding convex optimization problems for our new group Stepdown SLOPEs, including group $k$-SLOPE (g$k$-SLOPE) and group F-SLOPE (gF-SLOPE). Notice that g$k$-SLOPE and gF-SLOPE approaches can control g$k$-FWER and gFDP, respectively.

The control of g$k$-FWER and gFDP is defined as follows.
\begin{definition}
Given $\alpha,\gamma\in(0,1)$, the g$k$-FWER and gFDP control are define as 
\begin{equation}
\mathrm{g}k-\mathrm{FWER}=\mathrm{Prob}(Vg\geq k)<\alpha, 
\label{gkFWER control}
\end{equation}
and 
\begin{equation}
P(\mathrm{gFDP}>\gamma)=\mathrm{Prob}(\frac{Vg}{max\{Rg,1\}}>\gamma)<\alpha.\label{gFDP control}
\end{equation}
\end{definition}

Suppose that the $m$-dimensional response vector $y$ is factually generated by a linear model in (\ref{groupform}). We denote the value of $||X_{I_i}\beta_{I_i}||_2$ as the influence on the of $i$-th group response variable. We think that group $i$ is closely relevant if and only if $||X_{I_i}\beta_{I_i}||_2>0$. Therefore, the task of identifying relevant groups can be transformed into finding which elements in $\| \beta\|_{X,I}=(||X_{I_1}\beta_{I_1}||_2,...,||X_{I_t}\beta_{I_t}||_2)^{T}$ are greater than 0. To estimate the nonzero coefficients of $\| \beta\|_{X,I}$, we built two group stepdown SLOPE models, namely g$k$-SLOPE and gF-SLOPE, to control g$k$-FWER and gFDP, respectively.

The objective function of convex optimization w.r.t. g$k$-FWER is
\begin{align}
\mathop{\arg\min}_{\beta\in \mathbb{R}^m}\frac{1}{2}||y-X\beta||^2+\sigma J_{\lambda_{gk}}(W\| b\|_{X,I}),\label{gk-SLOPEeq}
\end{align}
where $X$ is an arbitrary matrix. When one considers $t$ groups containing only individual variables (i.e., singleton groups situation), all weights equal to one, and the object function (\ref{gk-SLOPEeq}) is simplified into (\ref{kFWER_model}). 
Then denote $\tilde{m}=l_1+l_2+...+l_t$ and consider the following partition $\mathbb{I}=\{\mathbb{I}_1,...,\mathbb{I}_t\}$ of $\{1,...,{t}\}$, we can derive
\begin{equation}
\begin{aligned}
&\mathbb{I}_1:=\{1,...,l_1\},\mathbb{I}_2:=\{l_1+1,...,l_1+l_2\},..., 
\text{and ~} \mathbb{I}_t:=\{\sum_{i=1}^{t-1}l_i+1,...,\sum_{i=1}^tl_i\}.
\end{aligned}
\end{equation}

We can find that $X_{I_i}$ can be decomposed into $X_{I_i}=U_iR_i$, where $U_i$ is a matrix with $l_i$ orthogonal column of a unit $l_2$ norm, whose span coincides with the space spanned by the columns of $X_{I_i}$ and $R_i$ is the corresponding matrix of a full row rank. Denote $n$ by $l$ matrix $\tilde{X}$ by setting $\tilde{X}_{\mathbb{I}_i}:=U_i$ and $c_{\mathbb{I}_i}:=R_i b_{I_i}$ for $i=1,...,t$. Then we have 
\begin{equation}
\begin{aligned}
Xb=\sum_{i=1}^t X_{I_i}b_{I_i}=\sum_{i=1}^t U_iR_ib_{I_i}=\sum_{i=1}^t\tilde{X}_{\mathbb{I}_i}c_{\mathbb{I}_i}=\tilde{X}b,
\end{aligned}
\end{equation}
and
\begin{equation}
\begin{aligned}
(b_{X,I})_i=||X_{I_i}b_{I_i}||_2=||R_ib_i||_2=||c_{\mathbb{I}_i}||_2.
\end{aligned}
\end{equation}

Then the problem (\ref{gk-SLOPEeq}) is equivalent to
\begin{equation}
\left\{\begin{array}{rl}
c^{\mathrm{gs}} & :=\underset{c}{\arg \min }\left\{\frac{1}{2}\|y-\widetilde{X} c\|_2^2+\sigma J_{\lambda_{gk}}\left(W \| c \|_{\mathbb{I}}\right)\right\} \\
c_{\mathbb{I}_i}^{\mathrm{gs}} & :=R_i \beta_{I_i}^{\mathrm{gs}}, i=1, \ldots, m
\end{array},\right.\label{sv_gk-SLOPEeq}
\end{equation}
where $\| c \|_{\mathbb{I}}:=\left(\left\|c_{\mathbb{I}_1}\right\|_2, \ldots,\left\|c_{\mathbb{I}_t}\right\|_2\right)^{\top}$. And (\ref{sv_gk-SLOPEeq}) can be regarded as the standard version of (\ref{gk-SLOPEeq}). 

In the following, we provide an intuitive theoretical analysis of the proposed two g-SLOPEs, namely g$k$-SLOPE and gF-SLOPE, which control the g$k$-FWER and gF-FWER criteria, respectively, under orthogonal group settings or, more generally, Gaussian settings.

\begin{algorithm}[!t]
\caption{Algorithms for g$k$-SLOPE}
\textbf{input}:$q\in(0,1),w_1,...,w_t > 0,p,n,m$, $l_1,...,l_m\in \mathbb{N}$;
$\lambda_i:=\overline{F}^{-1}(1-\frac{qi}{m})$ for $\overline{F}(x):=\frac{1}{m}\sum_{i=1}^mF_{w_i^{-1}\mathbf{X}_{l_i}}(x)$;

\begin{algorithmic}[0] %[1] enables line numbers
\FOR {$i\in\{2,...,m\}$}
\STATE $\lambda^S:=(\lambda_1,...,\lambda_{i-1})^{T}$
\STATE $\mathcal{S}_j:=\sqrt{\frac{n-l_j(i-1)}{n}+\frac{w_j^2\left\|\lambda^S\right\|_2^2}{n-l_j(i-1)-1}}, \quad$ for $j \in\{1, \ldots, m\}$
\STATE
if $i\leq k, \lambda_i^*:=\bar{F}_{\mathcal{S}}^{-1}\left(1-\frac{k\alpha}{2m}\right), \quad$\\
else $\lambda_i^*:=\bar{F}_{\mathcal{S}}^{-1}\left(1-\frac{k\alpha}{2(m+k-i)}\right), \quad$ where $\bar{F}_{\mathcal{S}}(x):=\frac{1}{m} \sum_{j=1}^m F_{\mathcal{S}_j w_i^{-1} \chi_{l_j}}(x)$
\STATE if $\lambda_i^*<=\lambda_{i-1}$, let $\lambda_i:=\lambda_{i}^*$. \\
Otherwise, \textbf{end for} and $\lambda_j:=\lambda_{i-1}$ for $j\geq i$
\ENDFOR
\end{algorithmic}
\label{algorithm2}
\end{algorithm}

Before the statement of the primary analysis on gk-FWER control, we recall the definition of $\mathcal{X}$ distribution and define a scaled $\mathcal{X}$ distribution.

\begin{definition}\citep{Groupslope}
A random variable $X_1$ is considered to be obeying $\mathcal{X}$ distribution with $l$ degrees of freedom, i.e. $X_1\sim \mathcal{X}_l$, when $X_1$ can be expressed as $X_1=\sqrt{X_2}$ for $X_2$ satisfying a $\mathcal{X}^2$ distribution with $l$ degrees of freedom. In addition, we will say that a random variable $X_1$ has a scaled $\mathcal{X}$ distribution with $l$ degrees of freedom and scale $S$, when $X_1$ could be expressed as $X_1:=S\cdot X_2$, for $X_2$ having a $\mathcal{X}$ distribution with $l$ degrees of freedom. For simplicity, we use the notation $X_1\sim S\mathcal{X}_l$.
\end{definition}

\begin{proposition}\label{gk-SLOPE_theorem}
In the linear model (\ref{groupform}) with t disjoint groups and the orthogonal matrix $X$, i.e. $X_{I_i}^{T}X_{I_j}=0$ for $i\neq j$. Define $w_1,...,w_t$ as the sizes of each group and given $k\in\{1,2,...,t\}$, the sequence of regularizing parameters can be expressed as 
\begin{equation} \label{sequence2}
\lambda_{gk}(i)= \begin{cases}\max_{j=1,...,t}{\frac{1}{w_i}F^{-1}_{X_{l_j}}(1-\frac{k\alpha}{2m})}, & i \leq k \\ \max_{j=1,...,t}{\frac{1}{w_i}F^{-1}_{X_{l_j}}(1-\frac{k\alpha}{2(m+k-i)})}, & i>k.\end{cases}
\end{equation}
where $F_{X_{l_j}}$ is a cumulative distribution function of $X$ distribution with $l_j$ degrees of freedom. Then any solution, $\beta^{gS}$, w.r.t. problem g$k$-SLOPE (\ref{gk-SLOPEeq}) generates the same vector $\|\beta^{gS}\|_{X,I}$. Given $\alpha\in(0,1)$, the g$k$-FWER control (\ref{gkFWER control}) holds.
\end{proposition}

Proposition \ref{gk-SLOPE_theorem} shows that when the regularization parameter value satisfies Eq.\eqref{sequence2}, g$k$-SLOPE can control g$k$-FWER, which is verified in subsequent data experiments. However, it is very challenging to make data groups orthogonal in practical applications. So we extend this model to more general situations. Similar to $k$-SLOPE, g$k$-SLOPE also faces the same problem, which may be addressed by a heuristic modification of $\lambda_{gk}$ in a general setting. This modified sequence was calculated under the assumption that the explanatory variables are randomly sampled from a Gaussian distribution. Combining this idea with \citep{Groupslope}, we propose Algorithm \ref{algorithm2} for calculating the sequence of tuning parameters when variables in different groups are independent.

The number of incorrectly selected features that people are willing to comply with is proportional to the number of features identified. Therefore, we may no longer focus on g$k$-FWER, but instead on gFDP. The objective function of gF-SLOPE is
\begin{align}
\mathop{\arg\min}_{\beta\in \mathbb{R}^m}\frac{1}{2}||y-X\beta||^2+\sigma J_{\lambda_{gF}}(W\| b\|_{X,I}).\label{gF-SLOPEeq}
\end{align}

Similar to (\ref{sv_gk-SLOPEeq}), the problem (\ref{gF-SLOPEeq}) can be equivalent to
\begin{equation}
\left\{\begin{aligned}
& c^{\mathrm{gs}}:=\underset{c}{\arg \min }\left\{\frac{1}{2}\|y-\tilde{X} c\|_2^2+\sigma J_{\lambda_{gF}}\left(W \| c \|_{\mathbb{I}}\right)\right\}, \\
& c_{\mathbb{I}_i}^{\mathrm{gs}}:=R_i \beta_{I_i}^{\mathrm{gs}}, i=1, \ldots, m.
\end{aligned}\right. \label{sv_gF-SLOPEeq}
\end{equation}

We can consider (\ref{sv_gk-SLOPEeq}) as the standard version of (\ref{gk-SLOPEeq}). In the following, we will present two sequences that provably control gk-FWER when the variables in different groups are mutually orthogonal or follow a Gaussian distribution. 

\begin{proposition}\label{gF-SLOPE_theorem}
In the linear model (\ref{groupform}) with the orthogonal matrix $X$, i.e. $X_{I_i}^{T}X_{I_j}=0$ for $i\neq j$. Define $w_1,...,w_t$ are positive numbers and given $\alpha,\gamma\in (0,1)$, the sequence of regularizing parameter can be expressed as
\begin{equation} \label{prop4}
\lambda_{gF}(i)^{\max }:=\max_{j=1, \ldots, m}\left\{\frac{1}{w_j} F_{\chi_{l_j}}^{-1}\left(1-\frac{(\lfloor\gamma i\rfloor+1) \alpha}{2(m+\lfloor\gamma i\rfloor+1-i)}\right)\right\},
\end{equation}
where $F_{X_{l_j}}$ is a cumulative distribution function of $X$ distribution with $l_j$ degrees of freedom. Then any solution ($\beta^{gS}$) w.r.t. problem gF-SLOPE (\ref{gF-SLOPEeq}) generates the same vector $\|\beta^{gS}\|_{X,I}$. Thus for given $\alpha\in(0,1)$, the gFDP control (\ref{gFDP control}) holds.
\end{proposition}

Proposition \ref{gF-SLOPE_theorem} illustrates that when the regularization parameter value satisfies Eq.\eqref{prop4}, gF-SLOPE can control gF-FWER, which is verified in subsequent data experiments. However, it is very challenging to make data groups orthogonal in practical applications. So we extend this model to more general situations. Similar to F-SLOPE, gF-SLOPE also faces these issues, which may be addressed by a heuristic modification of $\lambda_{gF}$ in a general setting. This modified sequence was calculated under the assumption that the explanatory variables are randomly sampled from a Gaussian distribution. Based on this strategy, we designed Algorithm 3 to compute the tuning parameter sequence when variables in different groups are mutually independent. Notice that the proofs for Propositions \ref{gk-SLOPE_theorem} and \ref{gF-SLOPE_theorem} are derived by integrating our proof strategy in Theorems \ref{theorem1}, \ref{theorem2}, \ref{theorem3}, \ref{our_theorem1} and \ref{our_theorem2} to \citep{Groupslope}, thus we omit them.

\begin{algorithm}[!t]
\caption{Algorithms for gF-SLOPE}
\textbf{input}:$\alpha,\gamma\in(0,1),w_1,...,w_t>0,p,n,m$, $l_1,...,l_m\in \mathbb{N}$; 
$\lambda_i:=\overline{F}^{-1}(1-\frac{(\lfloor\gamma i\rfloor+1) \alpha}{2(m+\lfloor\gamma i\rfloor+1-i)})$ for $\overline{F}(x):=\frac{1}{m}\sum_{i=1}^mF_{w_i^{-1}\mathbf{X}_{l_i}}(x)$;

\begin{algorithmic}[0] %[1] enables line numbers
\FOR {$i\in\{2,...,m\}$}
\STATE $\lambda^S:=(\lambda_1,...,\lambda_{i-1})^{T}$
\STATE $\mathcal{S}_j:=\sqrt{\frac{n-l_j(i-1)}{n}+\frac{w_j^2\left\|\lambda^S\right\|_2^2}{n-l_j(i-1)-1}}, \quad$ for $j \in\{1, \ldots, m\}$
\STATE $\lambda_i^*:=\bar{F}_{\mathcal{S}}^{-1}\left(1-\frac{(\lfloor\gamma i\rfloor+1) \alpha}{2(m+\lfloor\gamma i\rfloor+1-i)}\right), \quad$ for $\quad \bar{F}_{\mathcal{S}}(x):=\frac{1}{m} \sum_{j=1}^m F_{\mathcal{S}_j w_i^{-1} \chi_{l_j}}(x)$
\STATE if $\lambda_i^*<=\lambda_{i-1}$, let $\lambda_i:=\lambda_{i}^*$. \\
Otherwise, \textbf{end for} and $\lambda_j:=\lambda_{i-1}$ for $j\geq i$
\ENDFOR
\end{algorithmic}
\end{algorithm}

\section{Empirical Validation}\label{sec5}
This section verifies the performance on error control of our stepdown SLOPE algorithms, e.g., F-SLOPE, $k$-SLOPE, gF-SLOPE and g$k$-SLOPE.
All experiments are implemented in R and Python on a PC equipped with an Intel(R) Xeon(R) Platinum 8175M CPU and an NVIDIA RTX A6000 GPU. The reported results are the average values obtained after repeating each experiment 100 times.

\begin{table*}[!t]
\centering
\caption{Results for controlled feature selection under the orthogonal design (different $t$ and fixed $k=5$). }
\begin{tabular}{c|ccc|ccc}
\hline
\multirow{2}{*}{t} & \multicolumn{3}{c|}{$k$-SLOPE} & \multicolumn{3}{c}{F-SLOPE} \\
   & $\mathrm{Prob}(\mathrm{FDP}>\gamma)$     & FDR     & Power   & $\mathrm{Prob}(\mathrm{FDP}>\gamma)$     & FDR     & Power   \\ \hline
50                 & 0.001    & 0.006   & 1.000   & 0.003   & 0.007   & 1.000   \\
100                & 0.000    & 0.002   & 0.998   & 0.002   & 0.005   & 1.000   \\
200                & 0.001    & 0.001   & 1.000   & 0.000   & 0.007   & 1.000   \\
300                & 0.002    & 0.001   & 1.000   & 0.000   & 0.005   & 0.995   \\
400                & 0.000    & 0.001   & 0.995   & 0.001   & 0.004   & 0.994   \\
500                & 0.000    & 0.001   & 0.997   & 0.000   & 0.005   & 0.997   \\ \hline
\end{tabular}
\label{tab1}
\end{table*}

\subsection{Experiments of Orthogonal Design (Stepdown SLOPE)}

Inspired by \citep{slope, Groupslope}, we draw the design matrix $X=I_n$ with $n=1000$. Then, we simulate the response from the linear model \begin{equation}
y=X\beta+\epsilon,~~ \epsilon\sim N(0,I_n).
\end{equation}

The number of relevant features varies within $t\in\{50,100,200,300, 400,500\}$ and the nonzero regression coefficients are equal to $3\sqrt{2\log n}$. 
We set the target FDR level $\alpha = 0.1$ and $\gamma = 0.1$ for F-SLOPE, and set $k\in\{5,10,15,20,25,30\}$ and $\alpha = 0.1$ for $k$-SLOPE. Table \ref{tab1} reports the estimation of FDR, $\rm{Prob}\mathrm\{FDP\geq \gamma\}$ and power with 100 repetitions. Figure \ref{orth_kfwer_figure} summarizes the results of SLOPE, $k$-SLOPE, and F-SLOPE in these trials. These results demonstrate that our proposed stepdown-based SLOPEs can achieve FDP, FDR, and $k$-FWER control flexibly, whereas the classical SLOPEs primarily focus on controlling the FDR criterion. Meanwhile, $k$-SLOPE and F-SLOPE also show strong performance across most empirical configurations. Furthermore, these empirical results verify the validity of Theorems 4 and 5.

\begin{figure}[!t]
\centering
\includegraphics[width=0.9\linewidth]{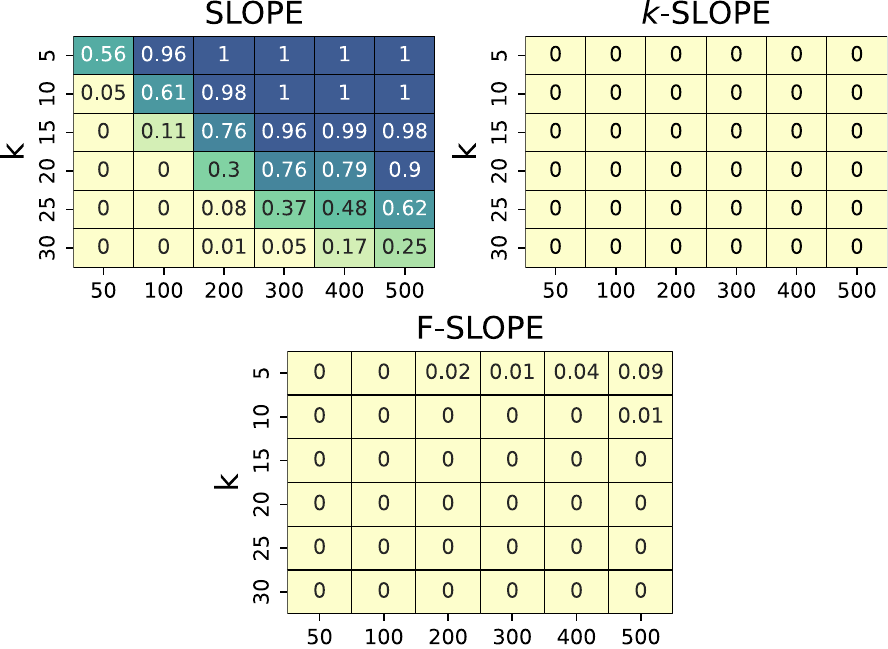}
\caption{$k$-FWER provided by different approaches for controlled feature selection under orthogonal design (with different $k$ and $t$). The value in the small square is the size of $k$-FWER. The darker the color, the larger the $k$-FWER and vice versa.}
\label{orth_kfwer_figure}
\end{figure}

\begin{figure*}[!t]
\centering
\includegraphics[width=\linewidth]{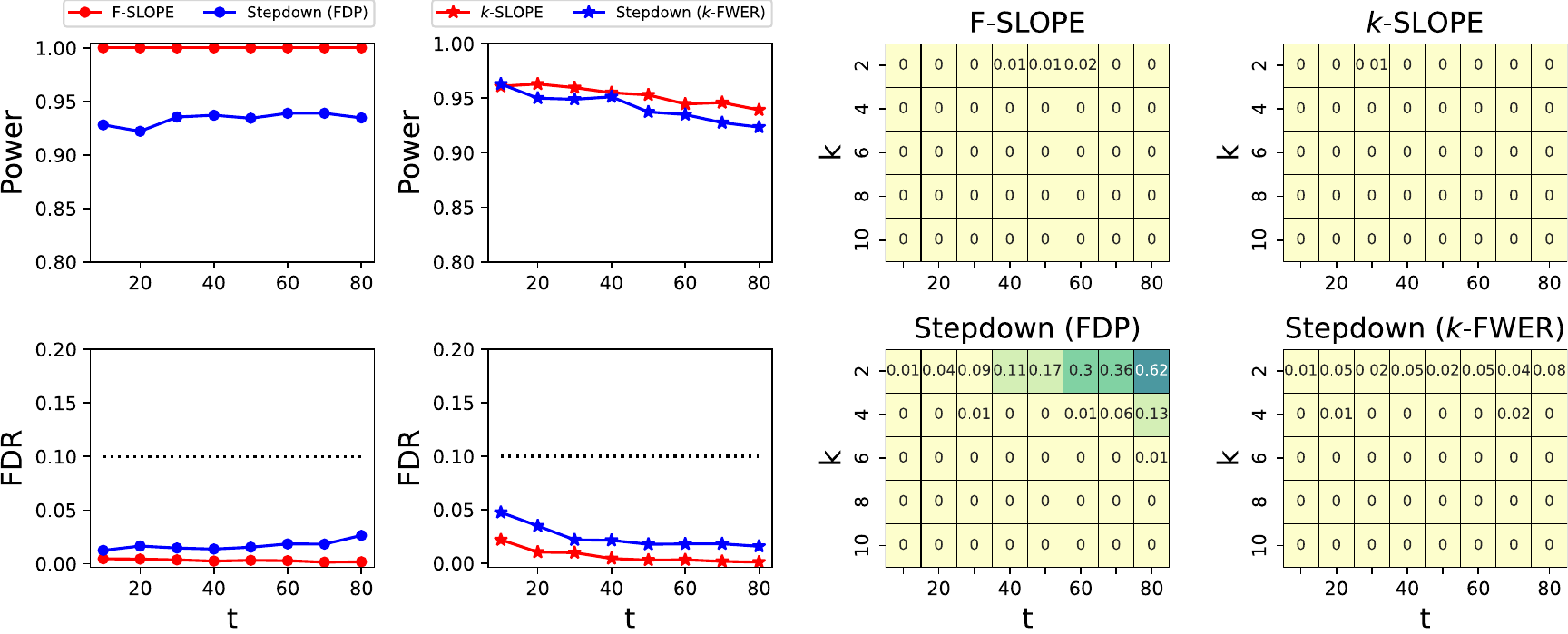}
\caption{Result for controlled feature selection on the simulated data. The black dashed lines indicate the target FDR level. Constance for $k$-SLOPE is $k=6$ in the second column (from left to right). The value in the small square is the size of $k$-FWER in the third and fourth columns (from left to right). The darker the color, the larger the $k$-FWER and vice versa. }
\label{figure2}
\end{figure*}

% \begin{figure}[!t]
% \centering
% \includegraphics[width=1\linewidth]{figure/figure3.pdf}
% \caption{ Power and FDR of F-SLOPE under Gaussian design (different $t$). The black dashed line indicates the target FDR level.}
% \label{figure3}
% \end{figure}

\begin{table*}[!t]
\centering
\caption{$\mathrm{Prob(FDP}>\gamma$) results on the simulated data for multiple mean testing ($k=6$)}\label{tab2}
% \resizebox{\linewidth}{!}{
\begin{tabular}{c|cccc}
\hline
$t$  & F-SLOPE & $k$-SLOPE & Sd (FDP) & Sd ($k$-FWER) \\ \hline
10 & 0.00    & 0.04    & 0.02       & 0.08            \\
20 & 0.03    & 0.00    & 0.00       & 0.03            \\
30 & 0.00    & 0.01    & 0.01       & 0.01            \\
40 & 0.00    & 0.00    & 0.00       & 0.00            \\
50 & 0.01    & 0.00    & 0.00       & 0.00            \\
60 & 0.00    & 0.00    & 0.00       & 0.00            \\
70 & 0.00    & 0.00    & 0.00       & 0.00            \\
80 & 0.00    & 0.00    & 0.00       & 0.00            \\ \hline
\end{tabular}
% }
\end{table*}

\begin{table*}[!t]
\centering
\caption{$k$-FWER of $k$-SLOPE ($m=2n$) on the simulated data under the weak signals }
% \resizebox{\linewidth}{!}{
\begin{tabular}{c|cccccccc}
\hline
$k/t$ & 10   & 20   & 30   & 40   & 50   & 60   & 70   & 80   \\ \hline
2   & 0.02 & 0.00 & 0.02 & 0.00 & 0.04 & 0.04 & 0.07 & 0.07 \\
4   & 0.00 & 0.00 & 0.00 & 0.01 & 0.00 & 0.00 & 0.00 & 0.01 \\
6   & 0.00 & 0.00 & 0.00 & 0.00 & 0.00 & 0.00 & 0.00 & 0.00 \\
8   & 0.00 & 0.00 & 0.00 & 0.01 & 0.00 & 0.00 & 0.00 & 0.01 \\ \hline
\end{tabular}
% }
\label{tab3}
\end{table*}

\begin{table*}[t]
\centering
\caption{FDR of F-SLOPE ($m=2n$) on the simulated data under the weak signals} 
% \resizebox{\linewidth}{!}{
\begin{tabular}{c|cccccccc}
\hline
$k/t$ & 10   & 20   & 30   & 40   & 50   & 60   & 70   & 80   \\ \hline
2   & 0.00 & 0.00 & 0.05 & 0.03 & 0.10 & 0.12 & 0.07 & 0.13 \\
4   & 0.00 & 0.00 & 0.00 & 0.01 & 0.00 & 0.00 & 0.01 & 0.06 \\
6   & 0.00 & 0.00 & 0.00 & 0.00 & 0.00 & 0.00 & 0.00 & 0.01 \\
8   & 0.00 & 0.00 & 0.00 & 0.01 & 0.00 & 0.02 & 0.00 & 0.00 \\ \hline
\end{tabular}
% }
\label{tab4}
\end{table*}

\subsection{Multiple mean testing from correlated statistics(Stepdown SLOPE)}

We exemplify the properties of our proposed methods as applied to the typical multiple testing problem with correlated test statistics. Similar to the settings in \citep{slope}, we consider the following experimental design. 
Researchers conduct $n=1000$ experiments in each of $p=5$ randomly selected laboratories. Observation results are modeled as
\begin{equation}
y_{i, j}=\mu_{i}+\tau_{j}+z_{i, j}, \quad 1 \leq i \leq n,1 \leq j \leq p,
\end{equation}
where $\tau_{j}\sim N(0,\sigma_{\tau}^2)$ are the laboratory impact factors, $z_{i, j}\sim N(0,\sigma_{z}^2)$ are the errors and are independent of each other. Our goal is to test whether $\mu_i$ is equal to 0, i.e. $H_i: \mu_i= 0,i = 1,2,...,n$. Averaging the observed values of 5 laboratories, we get the mean of the results 
\begin{equation}
\bar{y}_{i}=\mu_{i}+\bar{\tau}+\bar{z}_{i}, \quad 1 \leq i \leq n,
\end{equation}
where $\bar{y}=(\bar{y}_1,...,\bar{y}_n)^T$ is drawn independently from $\mathcal{N}(\mu,\Sigma)$, where $\Sigma_{i, i}=\frac{1}{5}\sigma_{\tau}^2=\rho$ and $\Sigma_{i, j}=\frac{1}{5}(\sigma_{\tau}^2+\sigma_{z}^2)=\sigma^2$ for $i\neq j$ \citep{slope}. 
The key problem is to determine whether the marginal means of a multivariate Gaussian correlation vector converge to zero.
One classical solution is to perform marginal tests with $\bar{y}$ statistic, which depends on the stepdown procedure to control $k$-FWER or FDP \citep{kFWER2}. In other words, we sort the $\bar{y}$ sequence with $|\bar{y}|_{(1)}\geq |\bar{y}|_{(2)}\geq\cdots |\bar{y}|_{(m)}$. Then we use the stepdown procedure with $k$-FWER or FDP critical values. Another solution is to ``whiten the noise", i.e., the regression equation is reduced to
\begin{equation}
\tilde{y}=\Sigma^{-1 / 2} \bar{y}=\Sigma^{-1/2} \mu+\epsilon,\label{whiten}
\end{equation}
where $\epsilon\sim N(0,I_p)$, $\Sigma^{-1/2}$ is the regression design matrix. 
Suppose $\Sigma^{-1/2}$ is close to the orthogonal matrix. In that case, the multiple test problem is transformed into the feature selection problem under the approximate orthogonal design, where $k$-SLOPE and F-SLOPE can provide better performance.

\begin{figure}[!t]
\centering
\includegraphics[width=1\linewidth]{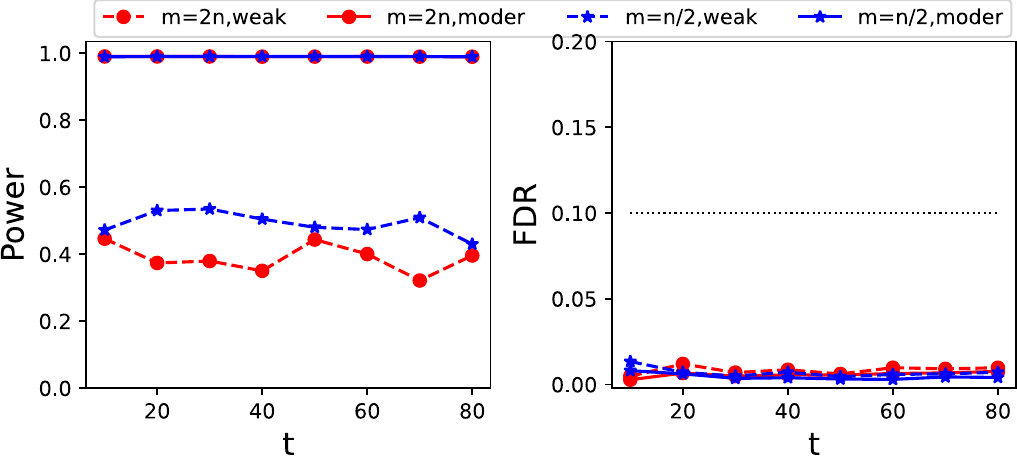}
\caption{Power and FDR of F-SLOPE under Gaussian design (different $t$). The black dashed line indicates the target FDR level.}
\label{figure3}
\end{figure}

\begin{figure}[!t]
\centering
\includegraphics[width=\linewidth]{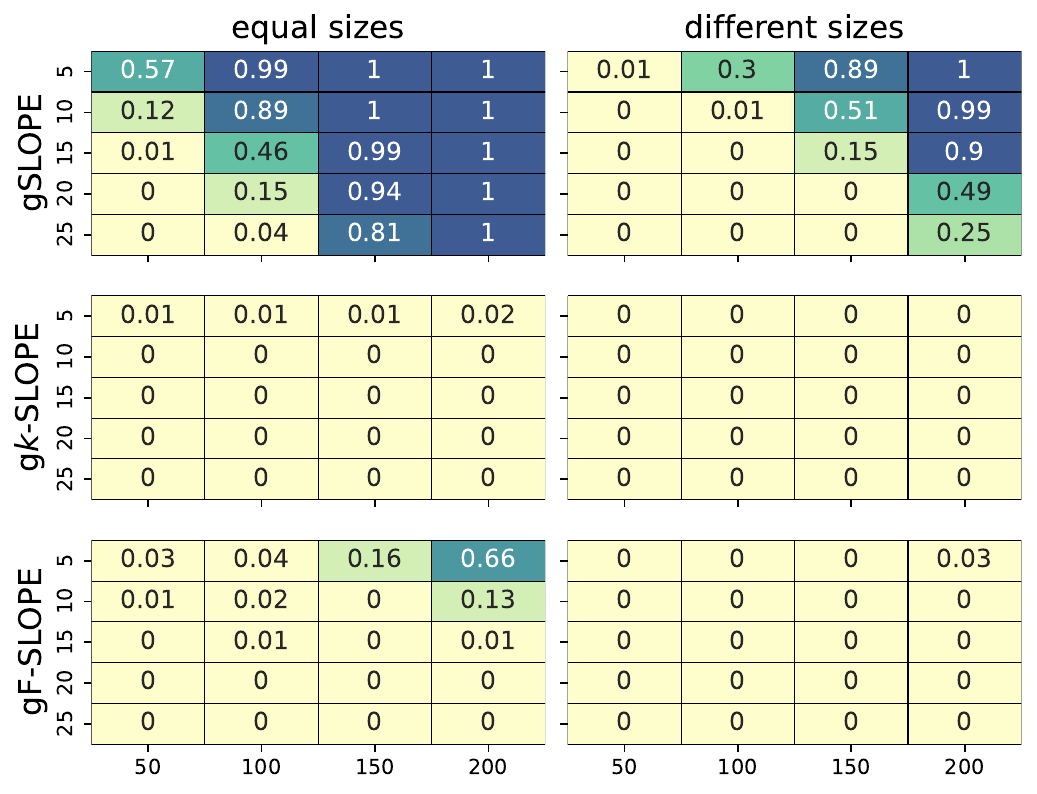}
\caption{The g$k$-FWER provided by different approaches for controlled feature selection under orthogonal design. The value in the small square is the size of g$k$-FWER. The darker the color, the larger the g$k$-FWER and vice versa.}
\label{figure7}
\end{figure}

Following the configurations in \citep{slope}, we let $\sigma_{\tau}^2=\sigma_{z}^2=2.5$ and consider a sparse setting where the number of the relevant features $t\in\{10,20,30,40,$ $50,60,70,80,90,100\}$. The nonzero mean is set to satisfy $2\sqrt{2\log{p}}/c$, where $c$ is equal to the Euclidean norm of each column of $\Sigma^{-1/2}$. We set $\alpha = \gamma = 0.1$ for all FDP controlled methods, and set $k\in\{2,4,6,8,10\}$ and $\alpha = 0.1$ for $k$-FWER controlled methods. Figure \ref{figure2} shows the FDR, $k$-FWER, and power provided by F-SLOPE, $k$-SLOPE, the stepdown procedures for FDP control (Sd(FDP)), and the stepdown procedures for $k$-FWER control (Sd($k$-FWER)). Table \ref{tab2} shows $\mathrm{Prob(FDP}>\gamma)$ for F-SLOPE, $k$-SLOPE and the stepdown procedures. These experimental results demonstrate that our proposed methods ensure simultaneous control of FDP, FDR, and $k$-FWER, whereas the variant Sd (FDP) (or Sd ($k$-FWER)) focuses on controlling FDP (or $k$-FWER) and FDR. Moreover, the proposed SLOPE approaches, F-SLOPE and $k$-SLOPE, have greater power than the stepdown procedures.
Therefore, our proposed methods outperform the classical stepdown procedures in multiple tests.

\subsection{Experiments of Gaussian Design (Stepdown SLOPE)}

We study the performance of $k$-SLOPE and F-SLOPE under a general setting. Following the strategy in \citep{slope}, let the entries of the design matrix $X$ be i.i.d. drawn from the distribution of $\mathcal{N} (0,1/n)$ with $n=5000$. The number of relevant features $t$ varies within $\{10,20,30,40,50,60,70,80\}$. Moderate signals with nonzero coefficients are set accordingly to $2\sqrt{2\log m}$, while this value is set to $\sqrt{2\log m}$ for weak signals. We set $\alpha = \gamma = 0.1$ for F-SLOPE, and set $k=\{2,4,6,8,10\}$ and $\alpha = 0.1$ for $k$-SLOPE.

\begin{figure}[!t]
\centering
\includegraphics[width=\linewidth]{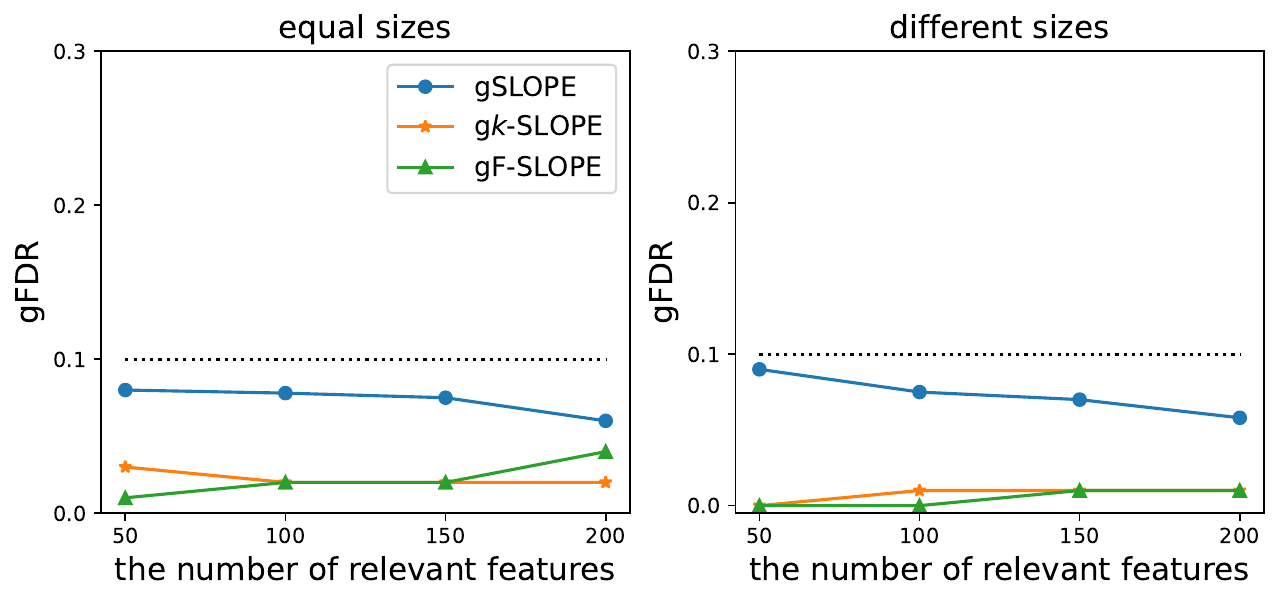}
\caption{The gFDR of g-SLOPE, g$k$-SLOPE and gF-SLOPE on the simulated data in orthogonal design when $k=15$.}
\label{figure8}
\end{figure}

\begin{table*}[!t]
\centering
\caption{\rm{Prob}(FDP$>\gamma$) of F-SLOPE on the simulated data under the weak and moderate signals (different $t$).}
% \resizebox{0.6\linewidth}{!}{
\begin{tabular}{c|cccc}
\hline
\multirow{2}{*}{$t$} & \multicolumn{2}{c}{$m = 2n$} & \multicolumn{2}{c}{$m = n/2$} \\
& weak       & moder       & weak        & moder       \\ \hline
10                 & 0.03       & 0.00        & 0.05        & 0.00        \\
20                 & 0.07       & 0.00        & 0.03        & 0.01        \\
30                 & 0.01       & 0.00        & 0.00        & 0.00        \\
40                 & 0.00       & 0.01        & 0.00        & 0.00        \\
50                 & 0.00       & 0.00        & 0.00        & 0.00        \\
60                 & 0.00       & 0.00        & 0.00        & 0.00        \\
70                 & 0.01       & 0.00        & 0.00        & 0.01        \\
80                 & 0.00       & 0.02        & 0.00        & 0.00        \\ \hline
\end{tabular}
% }
\label{tab5}
\end{table*}

\begin{table*}[!t]
\centering
\caption{\rm{Prob}(FDP$>\gamma$) of g-SLOPE, g$k$-SLOPE and gF-SLOPE on the simulated data.}
\begin{tabular}{c|cccc}
\hline
the number & 50         & 100     & 150       & 200     \\ \hline
g-SLOPE                          & 0.51 & 0.73 & 0.98 & 1.00 \\
g$k$-SLOPE                        & 0.01          & 0.00       & 0.00         & 0.01       \\
gFSLOPE                & 0.01         & 0.00       & 0.00         & 0.00      \\ \hline
\end{tabular}
\label{tab6}
\end{table*}

\begin{table*}[!t]
\centering
\caption{Power values of g-SLOPE, g$k$-SLOPE and gF-SLOPE in orthogonal design.}
\begin{tabular}{c|ccccc}
\hline
the number & 50  & 100 & 150  & 200  & 250  \\ \hline
g-SLOPE     & 0.8 & 0.9 & 0.93 & 0.95 & 0.96 \\
g$k$-SLOPE   & 0.99   & 0.98   & 1    & 0.99    & 1    \\
gF-SLOPE   & 1   & 0.99   & 1    & 1    & 0.98    \\ \hline
\end{tabular}
\label{tab7}
\end{table*}

Then we consider two scenarios: (1) $m=2n$; (2) $m=n/2$. Table \ref{tab4} and Figure \ref{figure3} illustrate that F-SLOPE keeps the $\mathrm{Prob(FDP}>\gamma)$ and FDR below the nominal level under both scenarios ($m=2n$ and $m=n/2$), whether the signals are weak and moderate. Meanwhile, Figure \ref{figure3} also shows F-SLOPE ($m=n/2$) has greater power than F-SLOPE ($m=2n$) under weak signals, while F-SLOPE ($m=2n$) and F-SLOPE ($m=n/2$) have similar power under the moderate signals. As shown in Table \ref{tab3}, $k$-SLOPE control $k$-FWER under both scenarios ($m=2n$ and $m=n/2$). In addition, the power of $k$-SLOPE also performs well with moderate signals. Moreover, experimental results verify the effectiveness of $k$-SLOPE with $\lambda_{k\mathrm{G}^*}$ and F-SLOPE with $\lambda_{\mathrm{FG}^*}$. 
\subsection{Experiments of Orthogonal Design 
(Group-wise Stepdown SLOPE)}

\begin{figure}[!t]
\centering
\includegraphics[width=\linewidth]{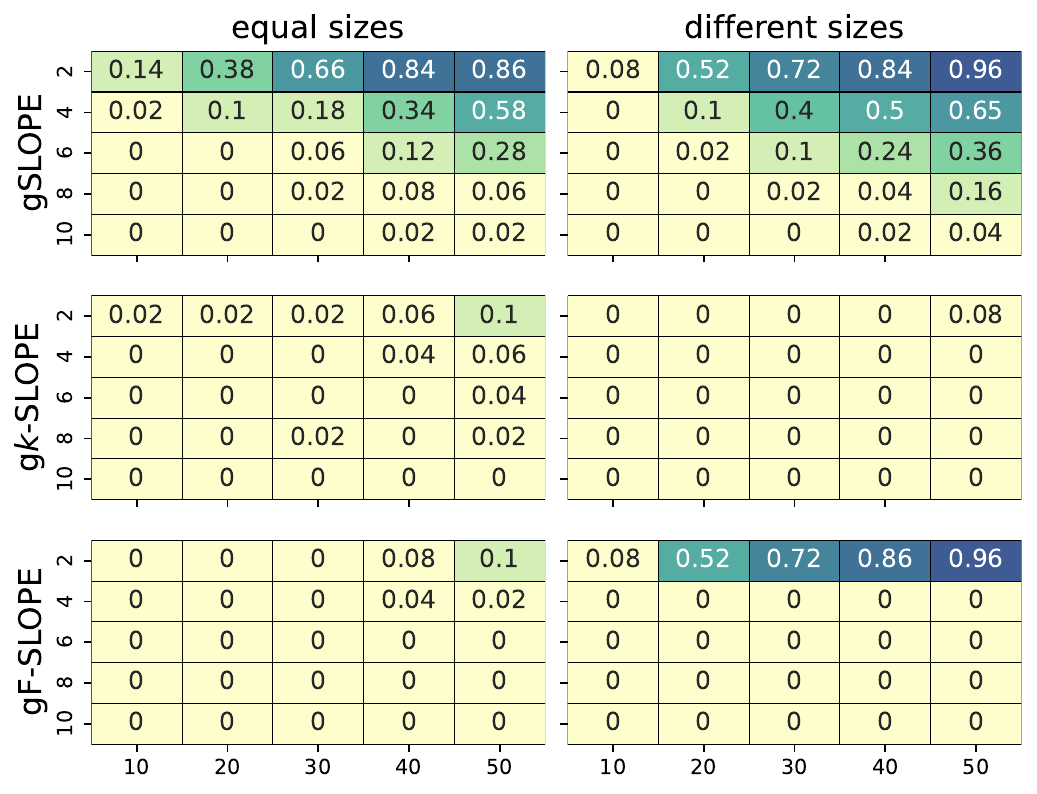}
\caption{The g$k$-FWER provided by different approaches for controlled feature selection under orthogonal design. The value in the small square is the size of g$k$-FWER. The darker the color, the larger the g$k$-FWER and vice versa.}
\label{figure9}
\end{figure}

Inspired by \citep{Groupslope}, we have designed the following experiments under orthogonal settings. Let the design matrix $X$ be the $p$ identity matrix, i.e., $X=I_p$ with $p=5000$. Within each important feature group, each element in the $\beta_{I_i}$ is generated independently under uniform distribution of $U[0.1,1.1]$, and the $\beta_{I_i}$ is scaled, that is, $(\beta_{X,I_i})=a\sqrt{l_i}$, where $a$ needs to satisfy $$\frac{1}{t} \sum_{i=1}^t a \sqrt{l_i}=\frac{1}{t} \sum_{i=1}^t \sqrt{4 \ln (t) /\left(1-t^{-2 / l}\right)-l}.$$ 
The number of relevant features is set to vary within $\{50, 100, 200, 300, 400, 500\}$. Then, we simulate the response from the linear model 
\begin{equation}
y=\sum_{i=1}^t X_{I_i}\beta_{I_i}+\epsilon,~~ \epsilon\sim N(0,I_p).\label{experiment1}
\end{equation}

For group settings, we consider two scenarios:

(1) There exist 1000 groups of explanatory variables with the exact group sizes of $l = 5$;

(2) There exist 1000 groups, which are equally divided into five categories with different group sizes of $\{3, 4, 5, 6, 7\}$.

We compare the three methods of group SLOPE, g$k$-SLOPE and gF-SLOPE under orthogonal conditions. We set the target FDR level $\alpha= 0.1$ and $\gamma = 0.1$ for gF-SLOPE, and $k = \{5, 10, 15, 20, 25, 30\}$ and $\alpha = 0.1$ for g$k$-SLOPE.

Figures \ref{figure7} and \ref{figure8} illustrate that although g-SLOPE, gk-SLOPE, and F-SLOPE are all successful in managing gFDR, gk-SLOPE surpasses the others in regulating the g$k$-FWER criterion, while g-SLOPE distinctly lacks this potential. Table \ref{tab6} shows that gk-SLOPE and gFSLOPE outperform g-SLOPE in controlling gFDP when $k=15$. In Table \ref{tab8}, it is evident that k-SLOPE and F-SLOPE demonstrate more power than g-SLOPE irrespective of the number of crucial features. 
In short, our proposed approaches are validated to be significantly advantageous in managing g$k$-FWER, gFDP, and power, thereby minimizing false feature selection and maximizing the selection of authentic features, particularly in the orthogonal scenario.

\begin{figure}[!t]
\centering
\includegraphics[width=\linewidth]{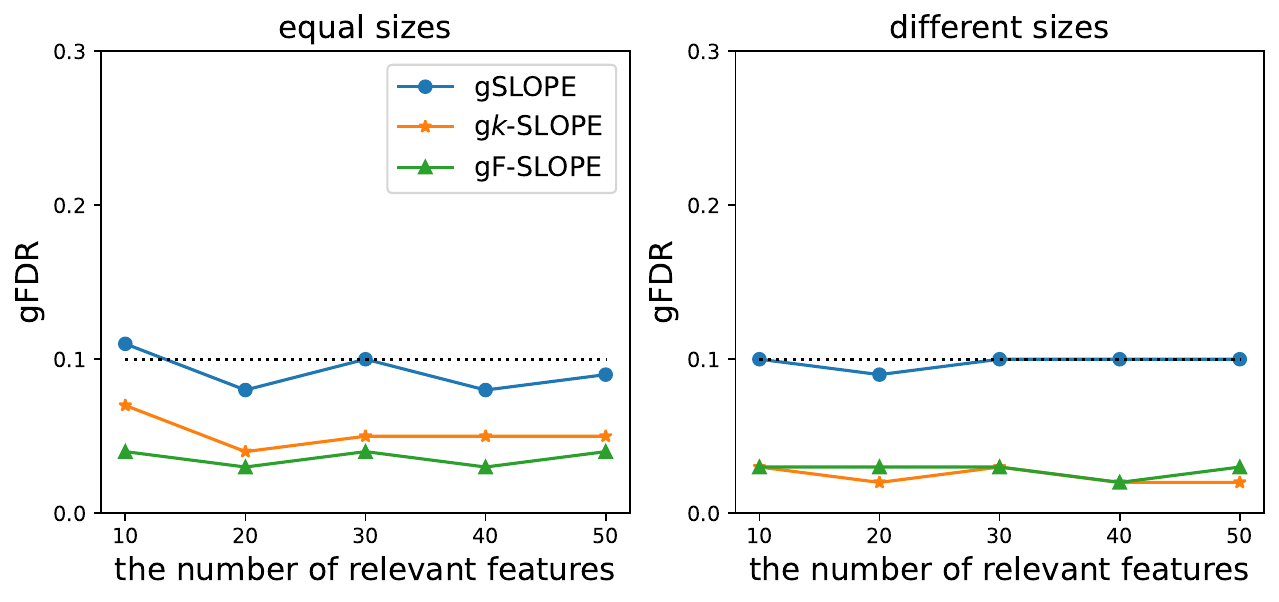}
\caption{The gFDR of g-SLOPE, g$k$-SLOPE and gF-SLOPE on the simulated data as $k=6$.}
\label{figure10}
\end{figure}

\begin{table*}[!t]
\centering
\caption{Results of \rm{Prob}(FDP$>\gamma$) on the simulated data. The values in different group configurations are shown in parentheses.}
% \resizebox{\linewidth}{!}{
\begin{tabular}{c|ccccc}
\hline
Size & 10         & 20         & 30         & 40         & 50         \\ \hline
g-SLOPE                         & 0.42(0.32) & 0.28(0.38) & 0.46(0.48) & 0.24(0.44) & 0.48(0.56) \\
gk-SLOPE                        & 0.16(0.00)    & 0.08(0.00)    & 0.12(0.02) & 0.02(0.00)    & 0.1(0.00)     \\
gF-SLOPE                        & 0.16(0.06) & 0.08(0.00)    & 0.12(0.00)    & 0.02(0.00)    & 0.1(0.00)     \\ \hline
\end{tabular}
% }
\label{tab8}
\end{table*}

\subsection{Experiments of Gaussian Design (Group Stepdown SLOPE)}

To evaluate the efficiency of g-SLOPE, g$k$-SLOPE, and F-SLOPE under a Gaussian design, we conducted simulations with a sample size of $n = 5000$ and $t = 1000$. For every level of sparsity and the gFDR, g$k$-FWER, and gFDP level of 0.1, during each iteration, we created entries of the design matrix using the Gaussian distribution $\mathcal{N}(0,1/n)$. Afterwards, we standardized $X$ and generated response variable values by model (\ref{experiment1}) using $\sigma= 1$. The signals were generated like the simulations presented in the orthogonal design. The number of relevant features varies within $\{10,20,30,40,50,60\}$. The grouping situation is the same as under an orthogonal design. We compare the three methods of group SLOPE, g$k$-SLOPE, and gF-SLOPE under Gaussian conditions.

Figures \ref{figure9} and \ref{figure10} illustrate that although g-SLOPE, gk-SLOPE, and F-SLOPE are all successful in managing gFDR, g$k$-SLOPE and gF-SLOPE surpass the others in regulating the g$k$-FWER criterion, while g-SLOPE distinctly lacks this potential. Table \ref{tab8} shows that gk-SLOPE and gFSLOPE outperform g-SLOPE in controlling gFDP when $k=6$. In Table \ref{tab9}, it is evident that k-SLOPE and F-SLOPE present more power than g-SLOPE irrespective of the number of crucial features. Briefly, the proposed new g-SLOPEs, including g$k$-SLOPE and gF-SLOPE, are shown to be significantly advantageous for controlling g$k$-FWER, gFDP, and power, thereby minimizing false feature selection and maximizing the selection of authentic features, particularly in the orthogonal scenario.

\begin{table*}[t]
\centering
\caption{Results of power values under general design.}
\begin{tabular}{c|ccccc}
\hline
Size & 50  & 100 & 150  & 200  & 250  \\ \hline
g-SLOPE     & 0.68 & 0.87 & 0.9 & 0.92 & 0.93 \\
g$k$-SLOPE   & 1   & 1   & 1    & 1    & 1    \\
gF-SLOPE   & 1   & 1   & 1    & 1    & 1    \\ \hline
\end{tabular}
\label{tab9}
\end{table*}

\subsection{Application on Real-world ADNI data}

\begin{table*}[!t]
\centering
\caption{Results on real-world ADNI data.}
\begin{tabular}{c|ccccc}
\hline
Method & MSE ($\times10^{-2}$)  & Time (s)  & gk-FWER  & Prob(gFDP$\geq$0.1)  \\ \hline
g-SLOPE &  2.84 ± 0.21  & 2.34 ± 0.82 & 0.82 & 0.48 \\
Group Lasso     & 2.79 ± 0.19 & 0.82 ± 0.02 &  0.91 &  0.63 \\
g$k$-SLOPE   & 2.61 ± 0.15   & 1.51 ± 0.43   & 0.12  & 0.17    \\
gF-SLOPE   & 2.60 ± 0.14   & 1.59 ± 0.47   & 0.19   & 0.11    \\ \hline
\end{tabular}
\label{adni}
\end{table*}

We evaluate gk-SLOPE and gF-SLOPE on the Alzheimer’s Disease Neuroimaging Initiative (ADNI) clinical data (downloaded from \url{http://adni.loni.usc.edu}), which contains prior expert knowledge on group structures and crucial features. To formulate the group modeling framework, we define feature groups based on the variable's structural categories or prior domain knowledge. For the ADNI dataset, features are partitioned according to their prior known group structure \citep{liu2019fused, zhang2023group, zhou2025interpretable}. The dataset comprises 116 variable groups: 46 groups contain a single covariate (subcortical volume, SV), while the remaining 70 groups each include four covariates: cortical thickness average (TA), thickness standard deviation (TS), surface area (SA), and cortical volume (CV) \citep{chen2026maximum}. The inherent grouping structure provides a valuable framework for our analytical approach and serves as a fundamental organizational principle for the neuroimaging data utilized in this investigation.

Table \ref{adni} validates the effectiveness and flexibility of our proposals on unknown distributions of realistic applications.
Overall, the empirical evidence supports the theoretical claims. The proposed stepdown regularization scheme achieves stringent control of gk-FWER and gFDP without the apparent power degradation typically associated with conservative multiple-testing corrections.

\section{Conclusion} \label{sec6}

This paper formulated new feature selection models based on the SLOPE technique \citep{slope}. Unlike existing works on FDR control, the current models focus on $k$-FWER control and FDP control for feature selection. Furthermore, we extend the proposed approaches for selecting grouped features, thereby further improving the classical g-SLOPE with guarantees of $k$-FWER and FDP control. With the help of the stepdown procedure \citep{kFWER2}, we established their theoretical guarantees under the orthogonal design. This work substantially broadens the scope and utility of SLOPE-based methods for rigorous feature selection across diverse error-control criteria and group-structural constraints.

\bibliography{sn-bibliography}% common bib file
%% if required, the content of .bbl file can be included here once bbl is generated
%%\input sn-article.bbl

\end{document}